\documentclass{llncs}
\usepackage{fullpage}

\usepackage[table]{xcolor}
\usepackage{xspace}
\usepackage{array}
\usepackage{subfigure}
\usepackage{xcolor}
\usepackage{epsfig}
\usepackage{gensymb}
\usepackage[english]{babel}
\usepackage{graphicx}
\usepackage{amssymb}
\usepackage{multirow}
\usepackage{todonotes}
\usepackage{enumerate}
\usepackage{paralist}
\usepackage{amsmath}
\usepackage{cases}
\usepackage{siunitx}

\newcommand{\remove}[1]{}

\renewcommand{\int}{int}

\newcommand{\generalHeight}{$O(n^{0.576})$\xspace}
\newcommand{\generalArea}{$O(n^{1.576})$\xspace}

\newcommand{\completeArea}{$O(n^{1.149})$\xspace}

\renewenvironment{proof}
{{\bf Proof:}}{\hspace*{\fill}$\Box$\par\vspace{2mm}}

\newtheorem{open}{Open Problem}

\begin{document}
\title{On the Area Requirements of Planar Straight-Line\\Orthogonal Drawings of Ternary Trees\thanks{Partially supported by MIUR Project ``MODE'' under PRIN 20157EFM5C and by H2020-MSCA-RISE project 734922 – ``CONNECT''. This paper combines the results contained in the conference papers~\cite{cfp-arslodtt-18} and~\cite{f-so-07}.}}
\author{Barbara Covella, Fabrizio Frati, Maurizio Patrignani}
\institute{
Roma Tre University, Italy -- \email{\{covella,frati,patrigna\}@dia.uniroma3.it}\\
}
\maketitle

\begin{abstract}
In this paper, we study the area requirements of planar straight-line orthogonal drawings of ternary trees. We prove that every ternary tree admits such a drawing in sub-quadratic area.  

Further, we present upper bounds, the outcomes of an experimental evaluation, and a conjecture on the area requirements of planar straight-line orthogonal drawings of complete ternary trees.

Finally, we present a polynomial lower bound on the length of the minimum side of any planar straight-line orthogonal drawing of a complete ternary tree.

\end{abstract}
{\em Keywords:} Graph drawing; ternary tree; planar straight-line orthogonal drawing; area requirements.

\section{Introduction}\label{le:intro}
A {\em planar straight-line orthogonal} drawing of a graph represents each vertex as a point in the plane and each edge either as a horizontal or as a vertical straight-line segment, so that no two edges cross; see Fig.~\ref{fig:example}.

\begin{figure}[htb]
	\begin{center}
		\mbox{\includegraphics[scale=1.2]{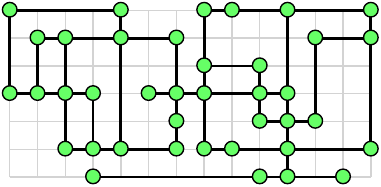}}
		\caption{A planar straight-line orthogonal drawing of a graph in which vertices are placed at grid points.}
		\label{fig:example}
	\end{center}
\end{figure}
Planar straight-line orthogonal graph drawings have long been studied. In 1987, Tamassia~\cite{t-egg-87} presented an algorithm that decides in polynomial time whether a graph with a fixed combinatorial embedding has a planar straight-line orthogonal drawing (and, more in general, a planar orthogonal drawing with at most $k$ bends, for any integer $k\geq 0$); this result lies at the very foundations of the research area now called Graph Drawing. On the other hand, Garg and Tamassia~\cite{gt-ccurpt-01} proved that deciding whether a graph admits a planar straight-line orthogonal drawing is NP-hard in the variable-embedding setting. Nomura et al.~\cite{ntu-odog-05} proved that every {\em outerplanar graph} with maximum degree $3$ and no $3$-cycle has a planar straight-line orthogonal drawing. 

The question whether a given {\em tree} has a planar straight-line orthogonal drawing is less interesting, as the answer is positive if and only if the degree of each node is at most $4$. Most research efforts concerning planar straight-line orthogonal drawings of trees have then been devoted to the construction of drawings with {\em small area}. This is usually formalized by requiring nodes to lie at {\em grid} points, i.e., points with integer coordinates, and by defining the {\em width} and {\em height} of a drawing as the number of grid columns and rows intersecting it, respectively, and the {\em area} as the width times the height. 

We introduce some definitions; see also~\cite{dett-gd-99,df-dt-13,r-tda-16}. A {\em rooted} tree $T$ is a tree with one distinguished node, which is called the {\em root} of $T$ and is denoted by $r(T)$. As usual in the literature about tree drawings, we will always assume trees to be rooted, even when not explicitly mentioned. Trees of maximum degree $3$ and $4$ are also called {\em binary} and {\em ternary} trees, respectively. For every node $s\neq r(T)$ in a tree $T$, the neighbor of $s$ in the path between $s$ and $r(T)$ is the {\em parent} of $s$; all the other neighbors of $s$ are its {\em children}. In a binary tree each node has at most $2$ children, while in a ternary tree each node has at most $3$ children; note that this requires the root to have degree at most $2$ or $3$, respectively. A {\em leaf} is a node with no children. For any non-leaf node $s$ of a tree $T$, removing $s$ and its incident edges from $T$ splits $T$ into several connected components; the ones containing children of $s$ are the {\em subtrees of $s$}; the subtrees of $s$ are rooted at the children of $s$. A {\em subtree} in a tree $T$ is either the tree $T$ itself, or a subtree of some non-leaf node $s$ in $T$. A tree is {\em complete} if every non-leaf node has the same number of children and every root-to-leaf path has the same length. We denote by $T_h$ the complete ternary tree such that every root-to-leaf path has length $h$, i.e., it consists of $h$ nodes. We denote by $|T|$ the number of nodes of a tree $T$.

It has been known since the 70's that $n$-node complete binary trees admit planar straight-line orthogonal drawings in $O(n)$ area~\cite{cdp-no-92,s-lpag-76}. Concerning general binary trees, $O(n\log \log n)$ has long stood as the best known area bound~\cite{cgkt-oa-02,skc-ae-00}; however, a recent breakthrough result of Chan~\cite{c-tdr-18} has improved that to $n2^{O(\log^* n)}$, where $\log^*$ denotes the iterated logarithm. 

In this paper we prove the first sub-quadratic area bound for planar straight-line orthogonal drawings of ternary trees. In fact, our main result is that every $n$-node ternary tree admits a planar straight-line orthogonal drawing in \generalArea area. In Section~\ref{se:general}, we present a recursive geometric construction, together with the (non-trivial) analysis of its area requirements, from which our main result follows. 

In Section~\ref{se:complete}, we study the area requirements of planar straight-line orthogonal drawings of $n$-node complete ternary trees. A recent result of Ali~\cite{a-so-15} proved that such drawings can be constructed in $O(n^{1.118})$ area. We focus on drawings that satisfy the {\em subtree separation property}: the smallest axis-parallel rectangles enclosing the drawings of any two node-disjoint subtrees do not overlap. We prove that $n$-node complete ternary trees have planar straight-line orthogonal drawings with the subtree separation property in $O(n^{1.149})$ area. We also present an algorithm that constructs a minimum-area planar straight-line orthogonal drawing with the subtree separation property of a complete ternary tree in polynomial time. This allowed us to experimentally compute the area required by planar straight-line orthogonal drawings with the subtree separation property for complete ternary trees with up to $2$ billion nodes. The outcomes of these experiments led us to conjecture that complete ternary trees do not admit planar straight-line orthogonal drawings with the subtree separation property in near-linear area.

Finally, in Section~\ref{se:lower-bound}, we prove that any planar straight-line orthogonal drawing of a complete ternary tree with $n$ nodes requires $\Omega(n^{\log_3 \phi})\in \Omega(n^{0.438})$ height {\em and} width, where $\phi=(1+\sqrt 5)/2$ is the golden ratio. This marks a notable difference between binary trees and ternary trees; in fact for the former planar straight-line orthogonal drawings can be constructed in which one side has logarithmic length, while for the latter our result proves that polynomial length might be required for both side lengths. Our lower bound holds true even for drawings that do not satisfy the subtree separation property. For $n$-node complete ternary trees, we prove that our bound is tight: they admit planar straight-line orthogonal drawings in which one side has length in $O(n^{\log_3 \phi})$.

Section~\ref{se:conclusions} concludes the paper with some open problems. In the remainder of the paper, a {\em drawing} will always mean a planar straight-line orthogonal drawing in which nodes lie at grid points.

\section{General Ternary Trees} \label{se:general}

In this section we prove the following theorem.

\begin{theorem} \label{th:general}
	Every $n$-node ternary tree admits a planar straight-line orthogonal drawing with width $O(n)$, with height \generalHeight, and with area \generalArea.
\end{theorem}

In the following we show an inductive algorithm that takes in input an $n$-node ternary tree $T$ and constructs a drawing $\Gamma$ of $T$ that satisfies the {\em top-visibility property}, i.e., the vertical half-line emanating from the root $r(T)$ and directed upwards does not intersect $\Gamma$, except at $r(T)$.

For a node $v$ in $T$, we denote as $T_v$ the subtree of $T$ rooted at $v$; further, we denote the subtrees of $v$ as the {\em heaviest} subtree $H_v$, the {\em second heaviest} subtree $M_v$, and the {\em lightest} subtree $L_v$ of $v$, according to the non-increasing order of the number of their nodes, with ties broken arbitrarily. A {\em heavy path} in $T$ is a path $(v_1,\dots,v_k)$ such that $r(T)=v_1$, such that $v_{i+1}$ is the root of $H_{v_i}$, for $i=1,\dots,k-1$, and such that $v_k$ is a leaf. We denote by $(\pi_1,\dots,\pi_{k(\pi)})$ the nodes of a heavy path $\pi$.

The base case of the construction is the one in which $n=1$. Then $\Gamma$ is trivially constructed by placing $r(T)$ at any grid point of the plane.

If $n>1$, then let $\pi$ be a heavy path in $T$. Further, let $\rho$ be a heavy path in $M_{r(T)}$. Let $p>4$ be a parameter to be determined later and let $x$ be the smallest index such that $\pi_x$ has at least two subtrees with at least $n/p$ nodes each, if any such an index exists. We first describe our construction by assuming that $x$ exists and is greater than $2$; we will deal with the other cases later. Let $\sigma$ be a heavy path in $M_{\pi_{x-1}}$ and let $\tau$ be a heavy path in $M_{\pi_x}$. Finally, let $P=(\rho_{k(\rho)},\dots,\rho_1,\pi_1,\dots,\pi_{x-1},\sigma_1,\dots,\sigma_{k(\sigma)})$ and  $Q=(\pi_{k(\pi)},\dots,\pi_x,\tau_1,\dots,\tau_{k(\tau)})$.

\begin{figure}[tb]
	\begin{center}
		\mbox{\includegraphics[width=.99\textwidth]{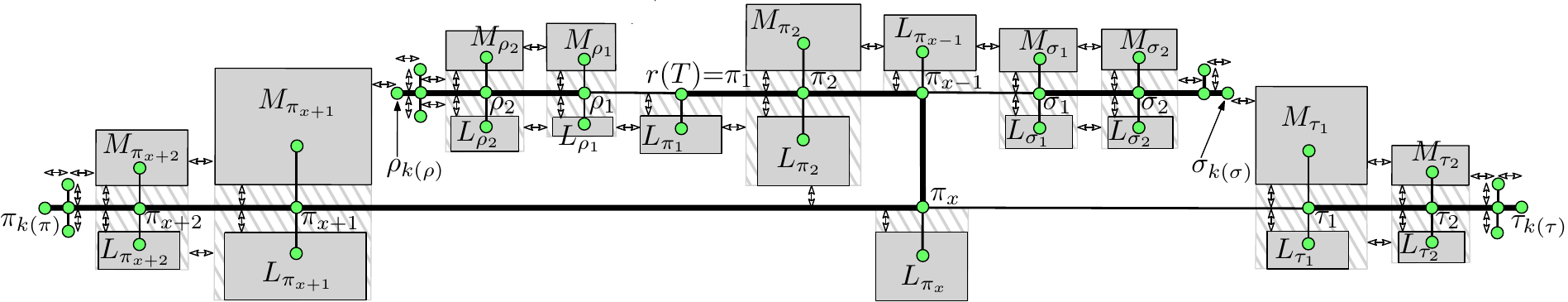}}
		\caption{Construction of $\Gamma$ if $n>1$. Fat lines represent $\pi$, $\rho$, $\sigma$, and $\tau$. Double-headed arrows indicate unit distances. Gray boxes represent inductively constructed drawings.}
		\label{fig:inductive-construction}
	\end{center}
\end{figure}

Fig.~\ref{fig:inductive-construction} shows the construction of $\Gamma$. The paths $P$ and $Q$ are drawn on two horizontal lines $\ell_P$ and $\ell_Q$, respectively, with $\ell_P$ above $\ell_Q$ and with the nodes in left-to-right order as they appear in $P$ and $Q$. Let ${\cal V}=V(P) \cup V(Q)$. For every subtree $T^*$ of $T$ rooted at a node $r^*\notin {\cal V}$ that is child of a node in ${\cal V}$, a drawing $\Gamma^*$ of $T^*$ is constructed inductively. Then $\Gamma^*$ is attached to the parent $p^*$ of $r^*$ as follows (note that $p^*\in{\cal V}$). If $T^*$ is the lightest subtree of $p^*$, then $\Gamma^*$ is placed with $r^*$ on the same vertical line as $p^*$ and with its top side one unit below $p^*$ (we call $T^*$ a {\em bottom subtree} of $P$ or $Q$, depending on whether $p^*$ is in $V(P)$ or $V(Q)$, respectively). Otherwise, $T^*$ is the second heaviest subtree of $p^*$ (note that $T^*$ is not the heaviest subtree of $p^*$, as $r^*\notin {\cal V}$); then $\Gamma^*$ is  rotated by $180\degree$ and placed with $r^*$ on the same vertical line as $p^*$ and with its bottom side one unit above $p^*$ (we call $T^*$ a {\em top subtree} of $P$ or $Q$, depending on whether $p^*$ is in $V(P)$ or $V(Q)$, respectively). There is one exception to this rule, which happens if  $T^*$ is the lightest subtree of $p^*=\pi_{x-1}$; then $\Gamma^*$ is rotated by $180\degree$ and placed with $r^*$ on the same vertical line as $p^*$ and with its bottom side one unit above $p^*$, as if it were a second heaviest subtree (then $T^*$ is a {\em top subtree} of $P$).

The horizontal distance of the nodes in $\cal V$ is determined so to ensure that, for any two nodes $x$ and $y$ such that $x$ comes immediately before $y$ on $P$ or $Q$, the rightmost vertical line intersecting $x$ or its attached subtrees is one unit to the left of the leftmost vertical line intersecting $y$ or its attached subtrees. There are two exceptions to this rule, involving the distance between $\pi_x$ and its children $\pi_{x+1}$ and $\tau_1$. Indeed, the distance between $\pi_x$ and $\pi_{x+1}$ is determined so that the rightmost vertical line intersecting $\pi_{x+1}$ or its attached subtrees is one unit to the left of the leftmost vertical line intersecting $P$, or its attached subtrees, or $\pi_x$, or its attached subtree $L_{\pi_x}$; the distance between $\pi_x$ and $\tau_1$ is determined similarly. The reason for ``pushing'' $\pi_{x+1}$ (resp.\ $\tau_1$) and its attached subtrees to the left (resp.\ to the right) of $P$ and its attached subtrees is to allow for a vertical compaction of $\Gamma$. In fact, the vertical distance between $P$ and $Q$ can now be chosen so that the bottommost horizontal line intersecting $P$ or its attached subtrees is one unit above $Q$. This completes the construction of $\Gamma$.

It is easy to see that $\Gamma$ is a planar straight-line orthogonal drawing satisfying the top-visibility property. Further, every grid column that intersects $\Gamma$ contains at least one node of $T$, hence the width of $\Gamma$ is in $O(n)$. We now analyze the height of $\Gamma$. Denote by $\eta(n)$ the maximum height of a drawing of a ternary tree with $n$ nodes constructed by the described algorithm. Then the height of $\Gamma$ is at most $\eta(n)$. Note that $\eta(1)=1$.

Let $a$ (resp.\ $b$) be the maximum number of nodes of a top (resp.\ bottom) subtree of $P$. Let $r$ (resp.\ $s$) be the maximum number of nodes of a top (resp.\ bottom) subtree of $Q$. By the definition of the index $x$ and since $|M_{\pi_i}|,|L_{\pi_i}|<|H_{\pi_i}|$, we have $|M_{\pi_i}|,|L_{\pi_i}|< n/p$, for any $i \in \{1,\dots,x-1\}$. Hence:
\begin{eqnarray} \label{eq:na-nb}
a,b< n/p.
\end{eqnarray} 

Since the lightest subtree of the root of a tree with $m$ nodes has less than $m/3$ nodes, we have
\begin{eqnarray} \label{eq:nd}
s\leq (n-a-b)/3.
\end{eqnarray} 

We also have the following inequality, whose proof needs some case analysis.
\begin{eqnarray} \label{eq:nc+nd}
r+s\leq \frac{2(p-1)}{3p}n.
\end{eqnarray} 

\begin{proof}
Let $R$ and $S$ be a top and bottom subtree of $Q$, respectively, with $|R|=r$ and $|S|=s$. By construction $R = M_{\pi_i}$, for some $x<i<k(\pi)$, or $R = M_{\tau_i}$, for some $1\leq i < k(\tau)$. Further, $S = L_{\pi_j}$, for some $x\leq j<k(\pi)$, or $S = L_{\tau_j}$, for some $1\leq j < k(\tau)$. We first assume that $R = M_{\pi_i}$, for some $x<i<k(\pi)$. We distinguish five cases.
	
\begin{figure}[htb]
	\begin{center}
		\begin{tabular}{c c c c c}
			\mbox{\includegraphics[scale = 0.9]{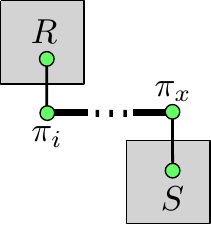}} & \hspace{3mm}
			\mbox{\includegraphics[scale = 0.9]{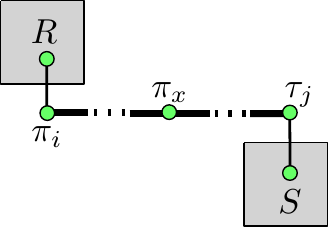}} & \hspace{3mm}
			\mbox{\includegraphics[scale = 0.9]{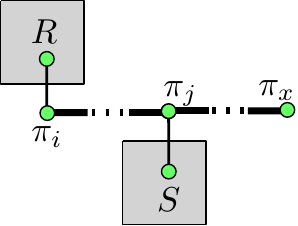}} & \hspace{3mm}
			\mbox{\includegraphics[scale = 0.9]{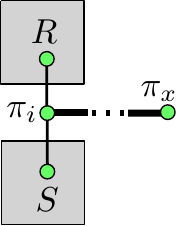}} & \hspace{3mm}
			\mbox{\includegraphics[scale = 0.9]{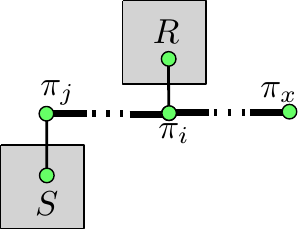}} \\
			(a) & \hspace{3mm}(b) & \hspace{3mm}(c)	 & \hspace{3mm}(d)	 & \hspace{3mm}(e)		
		\end{tabular}
		\caption{The five cases for the proof of inequality~(\ref{eq:nc+nd}).} 
		\label{fig:inequality3}
	\end{center}
\end{figure}

	{\em Case 1:} $S=L_{\pi_x}$ (see Fig.~\ref{fig:inequality3}(a)). We have that $|H_{\pi_i}|\geq r$; also, $|M_{\pi_x}|\geq s$. Since $L_{\pi_x}$, $M_{\pi_x}$, $M_{\pi_i}$, and $H_{\pi_i}$ are vertex-disjoint, we have $2r+2s\leq n$, which implies $r+s\leq \frac{2(p-1)}{3p}n$ if $p\geq 4$. 
	
	{\em Case 2:} $S=L_{\tau_j}$, for some $1\leq j < k(\tau)$  (see Fig.~\ref{fig:inequality3}(b)). We have $|H_{\pi_i}|\geq r$; also, $|H_{\tau_j}|\geq s$. Since $L_{\tau_j}$, $H_{\tau_j}$, $M_{\pi_i}$, and $H_{\pi_i}$ are vertex-disjoint, we have $2r+2s\leq n$, which implies $r+s\leq \frac{2(p-1)}{3p}n$ if $p\geq 4$. 
	
	{\em Case 3:} $S=L_{\pi_j}$, for some $x<j<i$  (see Fig.~\ref{fig:inequality3}(c)). We have $|H_{\pi_i}|\geq r$; also, $|M_{\pi_j}|\geq s$. Since $L_{\pi_j}$, $M_{\pi_j}$, $M_{\pi_i}$, and $H_{\pi_i}$ are vertex-disjoint, we have $2r+2s\leq n$, which implies $r+s\leq \frac{2(p-1)}{3p}n$ if $p\geq 4$.
	
	{\em Case 4:} $S=L_{\pi_i}$  (see Fig.~\ref{fig:inequality3}(d)). By definition of $x$ we have $|M_{\pi_x}|\geq n/p$. Since $T_{\pi_i}$ and $M_{\pi_x}$ are vertex-disjoint, we have $|T_{\pi_i}|\leq \frac{p-1}{p}n$. Since $|H_{\pi_i}|\geq |M_{\pi_i}|,|L_{\pi_i}|$, we have $r+s\leq \frac{2(p-1)}{3p}n$. 
	
	{\em Case 5:} $S=L_{\pi_j}$, for some $x<i<j<k(\pi)$  (see Fig.~\ref{fig:inequality3}(e)). As in Case~4, we have $|T_{\pi_i}|\leq \frac{p-1}{p}n$. Since $T_{\pi_{i+1}}=H_{\pi_i}$, we have $r\leq |T_{\pi_{i+1}}|$. Since $|H_{\pi_j}|,|M_{\pi_j}|\geq |L_{\pi_j}|$, we have $s\leq |T_{\pi_{i+1}}|/3$, hence $r+s\leq 4|T_{\pi_{i+1}}|/3$. On the other hand, at least $|T_{\pi_{i+1}}|-|L_{\pi_j}|\geq 2|T_{\pi_{i+1}}|/3$ nodes of $T_{\pi_{i+1}}$ do not belong to $R$ or $S$, hence $\frac{4|T_{\pi_{i+1}}|}{3}+\frac{2|T_{\pi_{i+1}}|}{3}=2|T_{\pi_{i+1}}|\leq \frac{p-1}{p}n$, which is $|T_{\pi_{i+1}}|\leq \frac{p-1}{2p}n$. It follows that $r+s\leq \frac{2(p-1)}{3p}n$. 
	
	This concludes the discussion if $R = M_{\pi_i}$, for some $x<i<k(\pi)$. Note that our arguments above do not make use of the fact that $|T_{\pi_{x+1}}|\geq |T_{\tau_1}|$ (which is true since $T_{\pi_{x+1}}=H_{\pi_x}$ and $T_{\tau_1}=M_{\pi_x}$). Hence, symmetric arguments handle the case in which $R = M_{\tau_i}$, for some $1\leq i < k(\tau)$. 
\end{proof}

The height of the part of $\Gamma$ below $\ell_Q$ is given by the height of a bottom subtree of $Q$. Further, since $\ell_Q$ is one unit below the bottommost horizontal line intersecting $P$ or its attached subtrees, the height of the part of $\Gamma$ above $\ell_Q$ is given by the maximum between the height of a top subtree of $Q$, and the height of a top subtree of $P$ plus the height of a bottom subtree of $P$ plus one (corresponding to $\ell_P$). Since the heights of a top subtree of $P$, of a bottom subtree of $P$, of a top subtree of $Q$, and of a bottom subtree of $Q$ are at most $\eta(a)$, $\eta(b)$, $\eta(r)$, and $\eta(s)$, respectively, by taking into account the grid row of $\ell_Q$ we get: 
\begin{eqnarray} \label{eq:h-equation}
\eta(n)\leq \max\{\eta(r)+\eta(s)+1,\eta(a)+\eta(b)+\eta(s)+2\}.
\end{eqnarray} 

%

We are going to inductively prove that
\begin{eqnarray} \label{eq:h-bound}
\eta(n)\leq 2 \cdot n^{c} -1, \textrm{ where } c=\frac{1}{\log_2{\frac{3p}{p-1}}}.
\end{eqnarray}  

Note that inequality~(\ref{eq:h-bound}) is trivially true if $n=1$. Now inductively assume that inequality~(\ref{eq:h-bound}) is true for all integer values of the variable less than $n$. By inequality~(\ref{eq:h-equation}), it suffices to prove that $\max\{\eta(r)+\eta(s)+1,\eta(a)+\eta(b)+\eta(s)+2\}\leq 2 \cdot n^{c} -1$. 

\begin{itemize}
\item First, we need to have $\eta(r)+\eta(s)+1\leq 2 \cdot n^{c} -1$. By induction, $\eta(r)+\eta(s)+1 \leq 2 \cdot r^{c} -1 + 2 \cdot s^{c} -1 +1$, hence we need that $r^{c} + s^c \leq n^c$. 

Here we use H\"older's inequality, which states that $\sum a_i b_i \leq (\sum a_i ^x)^{\frac{1}{x}}(\sum b_i ^y)^{\frac{1}{y}}$ for every $x$ and $y$ such that $1/x + 1/y=1$. By employing the values $1/x = c$, $1/y= 1-c$, $a_1=r^{c}$, $a_2=s^{c}$, $b_1=b_2=1$, we get:
\begin{eqnarray*} 
r^{c} + s^{c} \leq (r+s)^{c} \cdot 2^{1-c}\leq \left(\frac{2(p-1)}{3p}n\right)^{c} \cdot 2^{1-c}=2\cdot\left(\frac{p-1}{3p}n\right)^{c},
\end{eqnarray*}  
where we exploited inequality~(\ref{eq:nc+nd}). Thus, we need $2\cdot(\frac{p-1}{3p}n)^{c}\leq n^c$, which is $2\cdot (\frac{p-1}{3p})^{1/\left(\log_2{\frac{3p}{p-1}}\right)}\leq 1$. Set $x=\frac{3p}{p-1}$; then the previous inequality becomes $(1/x)^{1/\log_2 x}\leq 1/2$; taking the base-$2$ logarithms, we have $\log_2(1/x)^{1/\log_2 x}\leq \log_2(1/2)$, hence $\frac{1}{\log_2 x} \log_2(1/x) \leq -1$ which is $-1 \leq -1$. This proves that $\eta(r)+\eta(s)+1\leq 2 \cdot n^{c} -1$.  

\item Second, we need to have $\eta(a)+\eta(b)+\eta(s)+2\leq 2 \cdot n^{c} -1$. By inequality~(\ref{eq:nd}), we have $\eta(a)+\eta(b)+\eta(s)+2\leq \eta(a)+\eta(b)+\eta(\frac{n-a-b}{3})+2$. By induction, $\eta(a)\leq 2 \cdot a^c -1$, $\eta(b)\leq 2 \cdot b^c -1$, and $\eta(\frac{n-a-b}{3})\leq 2 \cdot (\frac{n-a-b}{3})^c -1$, hence we need that $a^c + b^c + (\frac{n-a-b}{3})^c \leq n^c$. 

We prove that $f(a,b)=a^c + b^c + (\frac{n-a-b}{3})^c$ grows monotonically with $a$, by assuming that $0.5<c<1$; this assumption will be verified later. We have $\frac{\partial f(a,b)}{\partial a}=c \cdot a^{c-1}-\frac{c}{3} \cdot(\frac{n-a-b}{3})^{c-1}$, which is greater than zero as long as $a^{c-1} > \frac{1}{3^c} \cdot(n-a-b)^{c-1}$. Since $c<1$, we have that $c-1$ is negative, hence by raising the previous inequality to the power of $1/(c-1)$ we get $a<3^{c/(1-c)}\cdot (n-a-b)$, which is $a<\frac{3^{c/(1-c)}}{1+3^{c/(1-c)}}\cdot (n-b)$. By inequality~(\ref{eq:na-nb}) the latter is true as long as $\frac{n}{p}<\frac{3^{c/(1-c)}}{1+3^{c/(1-c)}}\cdot \frac{p-1}{p} \cdot n$, that is $\frac{3^{c/(1-c)}}{1+3^{c/(1-c)}}\cdot (p-1)>1$. Since $c>0.5$ we get that $3^{c/(1-c)}>3$, hence $\frac{3^{c/(1-c)}}{1+3^{c/(1-c)}}>\frac{3}{4}$. Since $p>4$, the inequality $\frac{3^{c/(1-c)}}{1+3^{c/(1-c)}}\cdot (p-1)>1$ is satisfied, hence $\frac{\partial f(a,b)}{\partial a}>0$ and $f(a,b)$ grows monotonically with $a$.

It can be proved analogously that $f(a,b)$ grows monotonically with $b$, as long as $0.5<c<1$. 

By inequality~(\ref{eq:na-nb}) we have $a,b<n/p$, hence the monotonicity of $f(a,b)$ we proved above implies $a^c + b^c + \left(\frac{n-a-b}{3}\right)^c < 2\cdot (n/p)^c+ \left( \frac{n-2n/p}{3} \right)^c$. Thus, we need $2\cdot (n/p)^c+ \left( \frac{n-2n/p}{3} \right)^c\leq n^c$. Dividing by $n^c$, the inequality becomes $2\cdot (1/p)^c+ \left( \frac{1-2/p}{3} \right)^c-1 \leq 0$. Thus, we need to choose $p$ so to satisfy $2\cdot (1/p)^{1/\log_2{\frac{3p}{p-1}}}+ \left( \frac{1-2/p}{3} \right)^{1/\log_2{\frac{3p}{p-1}}}-1 \leq 0$; the latter inequality is true\footnote{We used the software at {\sc www.wolframalpha.com} in order to solve the inequality.} if $p\geq 9.956$. Thus setting $p=9.956$ we have $\eta(a)+\eta(b)+\eta(s)+2\leq 2 \cdot n^{c} -1$.
\end{itemize}

From $p=9.956$ we get $c=0.576$. By inequality~(\ref{eq:h-bound}) the height of $\Gamma$ is in  \generalHeight. This completes the proof of the height and area bounds for $\Gamma$.


Finally, we describe how to modify the construction if $x=1$, if $x=2$, or if $x$ is undefined. If $x$ is undefined, then $\pi$ ``never turns down'', that is, the construction coincides with the one above with $x=k(\pi)+1$. If $x=1$, then the construction coincides with the one above starting from $\pi_x$, that is, ignoring the paths $\rho$, $\sigma$, and $(\pi_1,\dots,\pi_{x-1})$ and their attached subtrees. If $x=2$, then $\pi$ ``immediately turns down'': the second heaviest subtree of $r(T)=\pi_1=\pi_{x-1}$ is drawn as above, while its lightest subtree is drawn as the second heaviest subtree of $\pi_{x-1}$ above ($\sigma$ is drawn straight with its subtrees attached to it); the rest of the construction, starting from $\pi_x$, coincides with the one above. In each of these cases, the analysis on the width, height, and area of the constructed drawings does not change. This concludes the proof of Theorem~\ref{th:general}.

\section{Complete Ternary Trees} \label{se:complete}

In this section we study the area requirements of planar straight-line orthogonal drawings of complete ternary trees. Restricting the attention to complete ternary trees allows one to achieve better area bounds than the one from Theorem~\ref{th:general}. Indeed, Ali~\cite{a-so-15} proved that  planar straight-line orthogonal drawings can be constructed in $O(n^{1.118})$ area\footnote{This result improved on an earlier $O(n^{1.262})$ bound~\cite{f-so-07}, which will be discussed later.}; the recursive construction achieving such a bound is depicted in Fig.~\ref{fig:ali}.

\begin{figure}[htb]
	\begin{center}
	\mbox{\includegraphics[scale = 0.75]{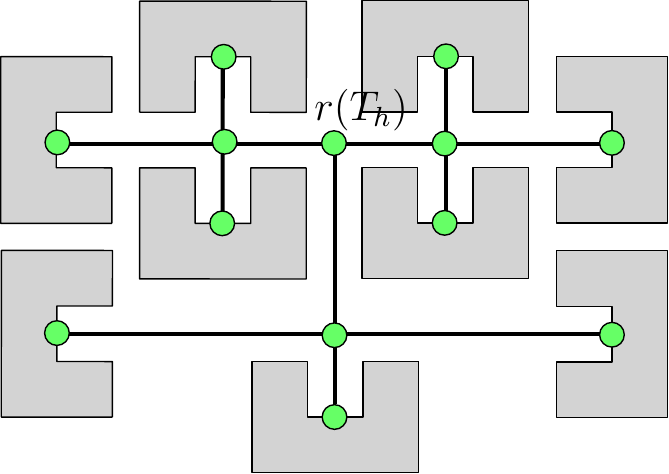}}
	\caption{Ali's~\cite{a-so-15} construction of a planar straight-line orthogonal drawing of the complete ternary tree $T_h$ with height $h$ uses $9$ copies of the inductively constructed drawing of $T_{h-2}$.} 
		\label{fig:ali}
	\end{center}
\end{figure}

We aim at investigating what area bounds can be achieved for planar straight-line orthogonal drawings of the complete ternary tree $T_h$ with height $h$ that are constructed by using (suitable combinations of) the two inductive constructions depicted in Figs.~\ref{fig:frati}(a) and~\ref{fig:frati}(b). These constructions are called Construction~1 and Construction~2, respectively. Before formally defining such constructions, we need to introduce the following definitions. For a drawing of a ternary tree, the {\em left width} is the number of grid columns intersecting the drawing to the left of the root; the {\em right width}, the {\em top height}, and the {\em bottom height} are defined analogously. 

\begin{figure}[tb]
	\begin{center}
		\begin{tabular}{c c}
			\mbox{\includegraphics[scale = 0.65]{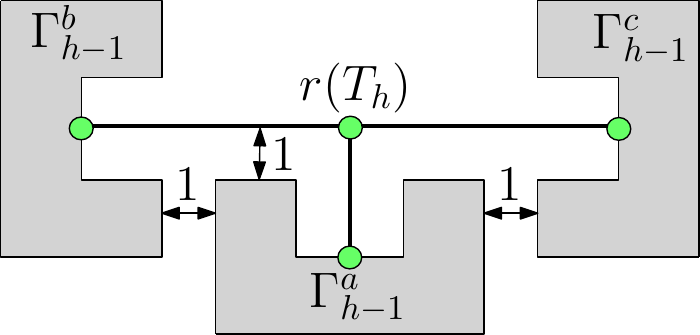}} & \hspace{8mm}
			\mbox{\includegraphics[scale = 0.65]{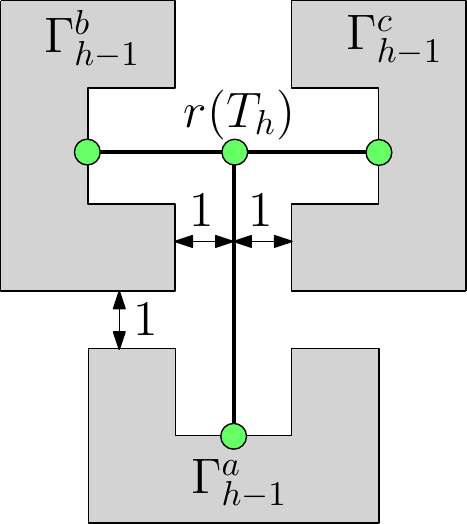}} \\
			(a) & \hspace{8mm}(b) 
		\end{tabular}
		\caption{(a) Construction~1 and (b) Construction~2. Each of them constructs a planar straight-line orthogonal drawing of $T_h$ out of $3$ (not necessarily congruent) 1-2 drawings $\Gamma^a_{h-1}$, $\Gamma^b_{h-1}$, and $\Gamma^c_{h-1}$ of $T_{h-1}$. A ``$\sqcup$''-shape is used to represent a 1-2 drawing to emphasize that the vertical half-line emanating from the root and directed upwards does not intersect the drawing other than at the root.} 
		\label{fig:frati}
	\end{center}
\end{figure}

A {\em 1-2 drawing} of $T_h$ is formally defined as follows. If $h=1$, then a drawing of $T_h$ in which the root $r(T_h)$ of $T_h$ is placed at a grid point of the plane is a 1-2 drawing; observe that this drawing is unique up to a translation by any vector with integer coordinates. If $h>1$, then consider any three (not necessarily congruent) 1-2 drawings of $T_{h-1}$, call them $\Gamma^a_{h-1}$, $\Gamma^b_{h-1}$, and $\Gamma^c_{h-1}$. Arrange such drawings as in Construction~1 or as in Construction~2; then the resulting drawing of $T_h$ is a 1-2 drawing. 

Construction~1 is more precisely defined as follows:

\begin{itemize}
\item draw $r(T_h)$ at any grid point of the plane;
\item place $\Gamma^a_{h-1}$ in the plane so that the highest horizontal line intersecting it is one unit below $r(T_h)$ and so that the root $r(T_{h-1})$ in $\Gamma^a_{h-1}$ is on the same vertical line as $r(T_h)$;
\item rotate $\Gamma^b_{h-1}$ in clockwise direction by $90\degree$ and then place it in the plane so that the rightmost vertical line intersecting it is one unit to the left of the leftmost vertical line intersecting $\Gamma^a_{h-1}$ and so that the root $r(T_{h-1})$ in $\Gamma^b_{h-1}$ is on the same horizontal line as $r(T_h)$; and
\item rotate $\Gamma^c_{h-1}$ in counterclockwise direction by $90\degree$ and then place it in the plane so that the leftmost vertical line intersecting it is one unit to the right of the rightmost vertical line intersecting $\Gamma^a_{h-1}$ and so that the root $r(T_{h-1})$ in $\Gamma^c_{h-1}$ is on the same horizontal line as $r(T_h)$.
\end{itemize}

The following descends immediately by construction.

\begin{property} \label{pr:construction-1}
Let $h\geq 2$ and, for $i\in \{a,b,c\}$, let $\omega_i$, $\eta_i$, $\lambda_i$, and $\rho_i$ be the width, the height, the left width and the right width of $\Gamma^i_{h-1}$, respectively. Then the width of the 1-2 drawing constructed by means of Construction~1 is $\omega_a+\eta_b+\eta_c$ and its height is $\max\{\lambda_b,\rho_c\}+\max\{\rho_b,\eta_a,\lambda_c\}+1$.
\end{property}

Construction~2 is analogously defined as follows:

\begin{itemize}
	\item draw $r(T_h)$ at any grid point of the plane;
	\item rotate $\Gamma^b_{h-1}$ in clockwise direction by $90\degree$ and then place it in the plane so that the rightmost vertical line intersecting it is one unit to the left of $r(T_h)$ and so that the root $r(T_{h-1})$ in $\Gamma^b_{h-1}$ is on the same horizontal line as $r(T_h)$; 
	\item rotate $\Gamma^c_{h-1}$ in counterclockwise direction by $90\degree$ and then place it in the plane so that the leftmost vertical line intersecting it is one unit to the right of $r(T_h)$ and so that the root $r(T_{h-1})$ in $\Gamma^c_{h-1}$ is on the same horizontal line as $r(T_h)$; and 
	\item place $\Gamma^a_{h-1}$ in the plane so that the highest horizontal line intersecting it is one unit below the lowest horizontal line intersecting $\Gamma^b_{h-1}$ or $\Gamma^c_{h-1}$ and so that the root $r(T_{h-1})$ in $\Gamma^a_{h-1}$ is on the same vertical line as $r(T_h)$.
\end{itemize}

The following descends immediately by construction.

\begin{property} \label{pr:construction-2}
	Let $h\geq 2$ and, for $i\in \{a,b,c\}$, let $\omega_i$, $\eta_i$, $\lambda_i$, and $\rho_i$ be the width, the height, the left width and the right width of $\Gamma^i_{h-1}$, respectively. Then the width of the 1-2 drawing constructed by means of Construction~2 is $\max\{\lambda_a,\eta_b\}+\max\{\rho_a,\eta_c\}+1$ and its height is $\max\{\lambda_b,\rho_c\}+\max\{\rho_b,\lambda_c\}+\eta_a+1$.
\end{property}

A 1-2 drawing has a nice feature that the drawings of Ali~\cite{a-so-15} do not have, called {\em subtree separation}  property: the smallest axis-parallel rectangles enclosing the drawings of any two node-disjoint subtrees do not overlap. This property has been frequently considered in the tree drawing literature (see, e.g.,~\cite{cgkt-oa-02,gr-sdbtla-02,r-tda-16,rs-gdbt-08}), both because of the readability of the drawings that have it and because it is directly guaranteed by the following natural approach for drawing trees: Inductively construct drawings of the subtrees of the root and place them together so that the smallest axis-parallel rectangles enclosing them do not overlap; a placement of the root in the plane then completes the drawing. The notorious {\em H-drawings}~\cite{s-lpag-76} and {\em HV-drawings}~\cite{cdp-no-92} are planar straight-line orthogonal drawings that satisfy the subtree separation property; they can be constructed in linear area for complete binary trees (see Figs.~\ref{fig:complete-old}(d) and~\ref{fig:complete-old}(e), respectively).  

\begin{figure}[htb]
	\begin{center}
		\begin{tabular}{c c c}
			\mbox{\includegraphics[scale = 0.6]{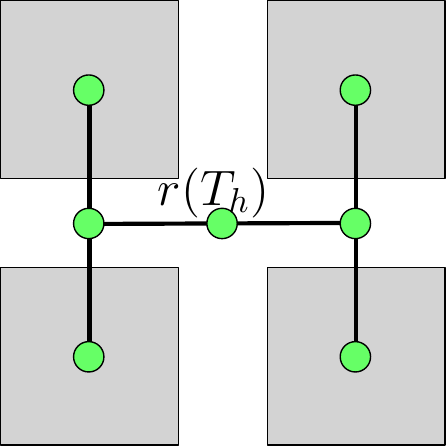}} & \hspace{8mm}
			\mbox{\includegraphics[scale = 0.6]{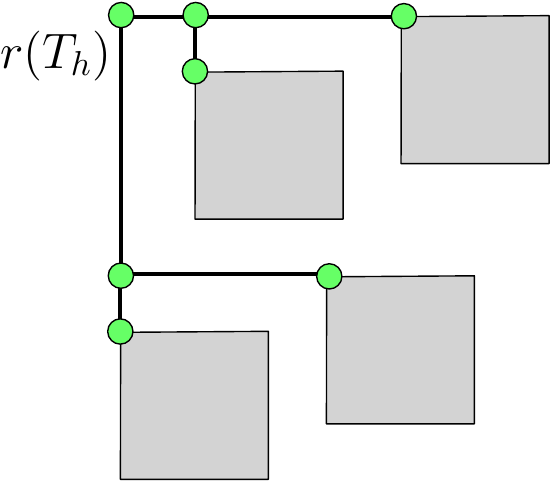}}\\
			(a) & \hspace{8mm}(b)\\
		\end{tabular}
		\caption{(a) H-drawings, as in~\cite{s-lpag-76}. (b) HV-drawings, as in~\cite{cdp-no-92}.} 
		\label{fig:complete-old}
	\end{center}
\end{figure}

In Section~\ref{se:12-optimal} we prove that, for complete ternary trees, 1-2 drawings require minimum area among all the planar straight-line orthogonal drawings satisfying the subtree separation property. Motivated by this result and exploiting a polynomial-time algorithm for computing minimum-area 1-2 drawings of complete ternary trees, in Section~\ref{se:experimental} we devise an extensive experimentation of the minimum area required by any 1-2 drawing of a complete ternary tree. However, we first analytically prove some asymptotic upper bounds for the area required by such drawings. This is done in the upcoming Section~\ref{se:complete-upper}.

\subsection{Upper Bounds for the Area Requirements of 1-2 Drawings} \label{se:complete-upper}

We start by considering 1-2 drawings that are constructed by using only one between Constructions~1 and~2. These drawings do not have the best area bounds, however they are natural to look at and easy to analyze. We have the following.

\begin{lemma} \label{le:construction-1-only}
Let $\Gamma_h$ be the 1-2 drawing of $T_h$ that is constructed by using Construction~1 only. Then $\Gamma_h$ has width $\omega_h=2^{h}-1$ and height $\eta_h=2^{h-1}$.
\end{lemma}

\begin{proof}
We prove, by induction on $h$, that $\Gamma_h$ has width $\omega_h=2^{h}-1$, has height $\eta_h=2^{h-1}$, and has the left width $\lambda_h$ equal to the right width $\rho_h$ (hence $\lambda_h=\rho_h=2^{h-1}-1$). 

The base case, in which $h=1$, is trivially proved, given that $\omega_1=\eta_1=1$ and $\lambda_1=\rho_1=0$. 

For the inductive case, assume that $\omega_{h-1}=2^{h-1}-1$, $\eta_{h-1}=2^{h-2}$, and $\lambda_{h-1}=\rho_{h-1}=2^{h-2}-1$. By Property~\ref{pr:construction-1} and since $\lambda_{h-1}=\rho_{h-1}$, we have $\omega_{h}= \omega_{h-1} + 2\cdot \eta_{h-1} = (2^{h-1}-1) + 2\cdot 2^{h-2}=2^h-1$ and $\eta_h=(2^{h-2}-1) + 2^{h-2} +1 = 2^{h-1}$. Finally, $\lambda_h=\eta_{h-1}+\lambda_{h-1}$ and $\rho_h=\eta_{h-1}+\rho_{h-1}$, hence $\lambda_h=\rho_h$. 
\end{proof}

\begin{lemma} \label{le:construction-2-only}
Let $\Gamma_h$ be the 1-2 drawing of $T_h$ that is constructed by using Construction~2 only. Then $\Gamma_h$ has:
\begin{itemize}
	\item width $\omega_h= (2^{h+1}-1)/3$ and height $\eta_h= (2^{h+1}-1)/3$ if $h$ is odd; and 
	\item width $\omega_h= (2^{h+1}+1)/3$ and height $\eta_h= (2^{h+1}-2)/3$ if $h$ is even.
\end{itemize} 
\end{lemma}

\begin{proof}
We prove, by induction on $h$, that $\Gamma_h$ has width and height as in the statement of the lemma, and that its left width $\lambda_h$ is equal to its right width $\rho_h$ (hence $\lambda_h=\rho_h=(2^h-2)/3$ if $h$ is odd and $\lambda_h=\rho_h=(2^h-1)/3$ if $h$ is even).

We prove the statement by induction on $h$. The base case, in which $h\leq 2$, is trivially proved, given that $\omega_1=\eta_1=1$, $\lambda_1=\rho_1=0$, $\omega_2=3$, $\eta_2=2$, and $\lambda_2=\rho_2=1$. 

For the inductive case, assume that the statement holds for $\Gamma_{h-1}$. By Property~\ref{pr:construction-2} and since $\eta_{h-1}>\lambda_{h-1}=\rho_{h-1}$, we have $\omega_h=\eta_{h-1} + \eta_{h-1} +1 = 2\cdot \eta_{h-1} + 1$ and $\eta_h= \lambda_{h-1}+\rho_{h-1}+\eta_{h-1}+1 = \omega_{h-1}+\eta_{h-1}$. 

\begin{itemize}
	\item If $h$ is odd, then $\omega_{h-1}= (2^{h}+1)/3$ and $\eta_{h-1}= (2^{h}-2)/3$, hence $\omega_h=2 \cdot (2^{h}-2)/3 +1 = (2^{h+1}-1)/3$ and $\eta_h= (2^{h}+1)/3 + (2^{h}-2)/3 = (2^{h+1}-1)/3$.
	\item If $h$ is even, then $\omega_{h-1}= (2^{h}-1)/3$ and $\eta_{h-1}= (2^{h}-1)/3$, hence $\omega_h= 2\cdot (2^{h}-1)/3 + 1 = (2^{h+1}+1)/3$ and height $\eta_h= (2^{h}-1)/3+ (2^{h}-1)/3 = (2^{h+1}-2)/3$.
\end{itemize}

Finally, since $\eta_{h-1}>\lambda_{h-1}=\rho_{h-1}$, we have $\lambda_h=\rho_h=\eta_{h-1}$. 
\end{proof}

\begin{corollary}
The 1-2 drawing of the complete ternary tree with $n$ nodes that is constructed by using Construction~1 only or by using Construction~2 only has width $O(n^{\log_3 2})\in O(n^{0.631})$, height $O(n^{\log_3 2})\in O(n^{0.631})$, and area $O(n^{\log_3 4})\in O(n^{1.262})$.
\end{corollary}

\begin{proof}
The statement follows by Lemma~\ref{le:construction-1-only}, by Lemma~\ref{le:construction-2-only}, and by $h=\log_3(2n+1)$.
\end{proof}

Our next algorithm constructs 1-2 drawings with height $O(n^{\log_3 \phi})\in O(n^{0.439})$, where $\phi$ is the golden ratio. It turns out that this height bound is the smallest possible. In fact, in Section~\ref{se:lower-bound} we will prove that every planar straight-line orthogonal drawing of a complete ternary tree with $n$ nodes requires height $\Omega(n^{\log_3 \phi})$, even if the drawing is not required to satisfy the subtree separation property.

\begin{theorem} \label{th:minimum-side-upper-bound}
The complete ternary tree with $n$ nodes has a 1-2 drawing with height $O(n^{\log_3 ((1+\sqrt 5)/2)})\in O(n^{0.439})$, with width $O(n^{\log_3 (1+\sqrt 2)})\in O(n^{0.803})$, and with area $O(n^{1.242})$.
\end{theorem}

\begin{proof}
In order to prove the theorem, we use induction on $h$. We will construct {\em two} 1-2 drawings $\Gamma^1_h$ and $\Gamma^2_h$ of $T_h$, the first one with small height and the second one with small width. 

If $h\leq 2$, then both $\Gamma^1_h$ and $\Gamma^2_h$ coincide with the unique 1-2 drawing of $T_h$. Suppose now that $\Gamma^1_{h-1}$ and $\Gamma^2_{h-1}$ have been defined, for some $h\geq 3$. We show how to construct $\Gamma^1_{h}$ and $\Gamma^2_{h}$.

We construct $\Gamma^1_h$ as follows. Let $\Gamma^a_{h-1}=\Gamma^1_{h-1}$ (this is the drawing of the subtree of $r(T_h)$ whose root is on the same vertical line as $r(T_h)$) and let $\Gamma^b_{h-1}=\Gamma^c_{h-1}=\Gamma^2_{h-1}$ (these are the drawings of the subtrees of $r(T_h)$ whose roots are on the same horizontal line as $r(T_h)$); recall that $\Gamma^b_{h-1}$ and $\Gamma^c_{h-1}$ are rotated by $90\degree$, respectively clockwise and counterclockwise, in $\Gamma^1_h$. We arrange such drawings together by means of Construction~1. Analogously, in order to construct $\Gamma^2_h$, let  $\Gamma^a_{h-1}=\Gamma^2_{h-1}$ and let $\Gamma^b_{h-1}=\Gamma^c_{h-1}=\Gamma^1_{h-1}$. We arrange such drawings together by means of Construction~2. 

We now analyze the height, the width, and the area of $\Gamma^1_h$ and $\Gamma^2_h$. Denote by $\omega^1_h$, by $\eta^1_h$, by $\lambda^1_h$, and by $\rho^1_h$ (by $\omega^2_h$, by $\eta^2_h$, by $\lambda^2_h$, and by $\rho^2_h$) the width, the height, the left width, and the right width of $\Gamma^1_h$ (of $\Gamma^2_h$), respectively. We will prove that: 

\begin{enumerate}[(a)]
	\item $\eta^1_{h} = \eta^1_{h-1}+\eta^1_{h-2}+1$, for $h\geq 3$;
	\item $\lambda^2_h=\rho^2_h$, for $h\geq 1$; and
	\item $\eta^1_{h} > \lambda^2_h$, for $h\geq 1$.
\end{enumerate}

We now prove the above inductive hypotheses by induction on $h$. The base case, in which $h\leq 2$, is trivially proved. Indeed, by construction we have $\lambda^2_1=\rho^2_1=0$ and $\lambda^2_2=\rho^2_2=1$, hence inductive hypothesis (b) is verified; further, $\eta^1_{1} = 1 > 0= \lambda^2_1$ and $\eta^1_{2} = 2 > 1 = \lambda^2_2$, hence inductive hypothesis (c) is verified. Inductive hypothesis (a) is vacuous if $h\leq 2$. Assume next that $h\geq 3$.

We prove that inductive hypothesis~(a) is verified for $h$. By Property~\ref{pr:construction-1} and by inductive hypotheses~(b) and~(c) for $h-1$, we have $\eta^1_h=\lambda^2_{h-1}+\eta^1_{h-1}+1$. By inductive hypothesis~(b) for $h-1$, we have $\lambda^2_{h-1}=(\omega^2_{h-1}-1)/2$. By Property~\ref{pr:construction-2} and by inductive hypothesis~(c) for $h-2$, we have $\omega^2_{h-1}=2\eta^1_{h-2}+1$, hence $\lambda^2_{h-1}=\eta^1_{h-2}$ and $\eta^1_h=\eta^1_{h-1}+\eta^1_{h-2}+1$.   

We prove that inductive hypothesis~(b) is verified for $h$. This comes from the fact that inductive hypothesis~(b) is satisfied for $h-1$ (which ensures that the left and right width of $\Gamma^a_{h-1}$ are equal) and from the fact that $\Gamma^b_{h-1}$ and $\Gamma^c_{h-1}$ have the same height (and hence the same width once they are rotated by $90\degree$). 

We prove that inductive hypothesis~(c) is verified for $h$. By construction, we have $\lambda^2_h=\max\{\eta^1_{h-1},\lambda^2_{h-1}\}$. By inductive hypothesis~(c) for $h-1$, we have $\eta^1_{h-1}>\lambda^2_{h-1}$ and hence $\lambda^2_h=\eta^1_{h-1}$. By inductive hypothesis~(a) for $h$, we have  $\eta^1_{h}>\eta^1_{h-1}$ and hence $\eta^1_{h}>\lambda^2_h$. 

We are now ready to analyze the height $\eta^1_h$ of $\Gamma^1_h$. By inductive hypothesis~(a), we have that $\eta^1_h$ grows asymptotically as the terms of the Fibonacci sequence; it is known that the ratio between two consecutive terms of such a sequence tends to the golden ratio $\phi=(1+\sqrt 5)/2 \approx 1.618$, hence $\eta^1_h\in O(\phi^h)$. This can be formalized as follows. We prove, by induction on $h$, that $\eta^1_h\leq k \cdot \phi^h -1$, where $k$ is a constant that is large enough so that $\eta^1_h\leq k \cdot \phi^h -1$ holds true in the base case, in which $h\leq 2$. If $h\geq 3$ then, by inductive hypothesis~(a), we have $\eta^1_h=\eta^1_{h-1}+\eta^1_{h-2}+1$. By induction, the latter is smaller than or equal to $(k \cdot \phi^{h-1} -1) + (k \cdot \phi^{h-2} -1) + 1 = k \cdot (\phi^{h-1} + \phi^{h-2}) -1$. Hence, it suffices to prove that $\phi^{h-1} + \phi^{h-2}\leq \phi^h$. Dividing both sides by $\phi^{h-2}$, we get $\phi^2-\phi-1 \geq 0$. Since $(1+ \sqrt 5)/2$ is one of the solutions of the equation $\phi^2-\phi-1 = 0$, it follows that $\phi^2-\phi-1 \geq 0$ and hence $\eta^1_h\leq k \cdot \phi^h -1\in O(\phi^h)$. Since $h\in O(\log_3 n)$ we get the height bound claimed in the theorem.

The width $\omega^1_h$ of $\Gamma^1_h$ needs to be analyzed together with the height $\eta^2_h$ of $\Gamma^2_h$. In fact, by construction and by Properties~\ref{pr:construction-1} and~\ref{pr:construction-2}, we have $\omega^1_h=\omega^1_{h-1}+2\cdot \eta^2_{h-1}$ and $\eta^2_h=\omega^1_{h-1}+\eta^2_{h-1}$.	Substituting repeatedly the first equation into the second one, we get $\eta^2_h=\eta^2_{h-1}+2\cdot \eta^2_{h-2}+2\cdot \eta^2_{h-3}+\cdots+2\cdot \eta^2_1 + \omega^1_1$. In the same way, we get $\eta^2_{h-1}=\eta^2_{h-2}+2\cdot \eta^2_{h-3}+2\cdot \eta^2_{h-4}+\cdots+2\cdot \eta^2_1 + \omega^1_1$. Subtracting the latter from the former, we get  $\eta^2_h-\eta^2_{h-1}=\eta^2_{h-1}+2\cdot \eta^2_{h-2} - \eta^2_{h-2}$, that is, $\eta^2_h=2\cdot \eta^2_{h-1}+\eta^2_{h-2}$. Now assume that $\eta^2_h \leq k\cdot c^h$, for some suitable constant $k$ and $c$, where we want $c$ as small as possible, while $k$ is sufficiently large so that $\eta^2_h \leq k\cdot c^h$ holds true for small values of $h$.  Then we need to have $2\cdot (k \cdot c^{h-1})+(k \cdot c^{h-2})\leq k \cdot c^h$, hence $c^2-2c-1\geq 0$. The associated equation has solutions $c=(1\pm \sqrt 2)$, hence we have $\eta^2_h \leq k \cdot (1 + \sqrt 2)^h \in O((1 + \sqrt 2)^h)$. The width $\omega^1_h$ of $\Gamma^1_h$ has the same asymptotic value of $\eta^2_h$. For example, from $\eta^2_{h+1}=\omega^1_{h}+\eta^2_{h}$ one derives $\omega^1_{h} \leq k \cdot (1 + \sqrt 2)^{h+1} \in O((1 + \sqrt 2)^h)$. By using $h\in O(\log_3 n)$ we get the width bound claimed in the theorem. The area bound also follows.
\end{proof}

The best area upper bound we could analytically prove for 1-2 drawings of $n$-node complete ternary trees is \completeArea; this is pretty close to the $O(n^{1.118})$ bound by Ali~\cite{a-so-15} for planar straight-line orthogonal drawings that do not satisfy the subtree separation property.

\begin{figure}[htb]
	\begin{center}
		\mbox{\includegraphics[scale = 0.75]{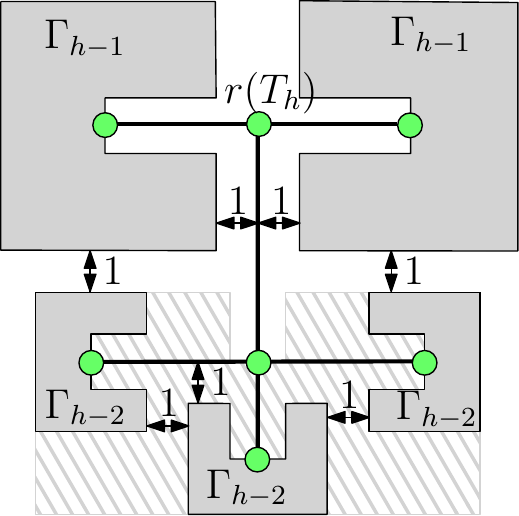}} 
		\caption{The construction of an \completeArea-area 1-2 drawing of $T_h$ satisfying the subtree separation property.} 
		\label{fig:complete}
	\end{center}
\end{figure}

\begin{theorem} \label{th:complete}
	The complete ternary tree with $n$ nodes has a 1-2 drawing in \completeArea area.
\end{theorem}

\begin{proof} 
	We show how to construct a 1-2 drawing $\Gamma_h$ of $T_h$ in \completeArea area. If $h\leq 2$, then $\Gamma_h$ coincides with the unique 1-2 drawing of $T_h$. Suppose now that $h\geq 3$. We construct $\Gamma_h$ as follows (refer to Fig.~\ref{fig:complete}). We construct $\Gamma^a_{h-1}$ (this is the drawing of the subtree of $r(T_h)$ whose root is on the same vertical line as $r(T_h)$) out of three copies of $\Gamma_{h-2}$ arranged by means of Construction~1. Further, let $\Gamma^b_{h-1}=\Gamma^c_{h-1}=\Gamma_{h-1}$ (these are the drawings of the subtrees of $r(T_h)$ whose roots are on the same horizontal line as $r(T_h)$); recall that $\Gamma^b_{h-1}$ and $\Gamma^c_{h-1}$ are rotated by $90\degree$, respectively clockwise and counterclockwise, in $\Gamma_h$. We arrange such drawings together by means of Construction~2. 
	 
	We now establish upper bounds on the width $\omega_h$ and on the height $\eta_h$ of $\Gamma_h$. In order to do that, we exploit that the left width $\lambda_h$ of $\Gamma_h$ is equal to the right width $\rho_h$. This is easily proved by induction. Indeed, in the base case, we have $\lambda_1=\rho_1=0$ and $\lambda_2=\rho_2=1$. Further, if $h\geq 3$, then $\lambda_h=\max\{\eta_{h-1},\omega_{h-2}+\lambda_{h-2}\}$ and $\rho_h=\max\{\eta_{h-1},\omega_{h-2}+\rho_{h-2}\}$, hence $\lambda_h=\rho_h$ follows by $\lambda_{h-2}=\rho_{h-2}$. 
	
	We are now ready to analyze the asymptotic behavior of $\omega_h$ and $\eta_h$. In the base case, we have $\omega_1=\eta_1=1$, $\omega_2=3$, and $\eta_2=2$. By Properties~\ref{pr:construction-1} and~\ref{pr:construction-2} and since $\lambda_h=\rho_h$, we have:
	\begin{eqnarray} \label{eq:w-complete}
	\omega_h= \max\{2\eta_{h-1}+1, \omega_{h-2} + 2\eta_{h-2}\}. 
	\end{eqnarray} 
	
	Again by Properties~\ref{pr:construction-1} and~\ref{pr:construction-2} and since $\lambda_{h-1}=\rho_{h-1}$ and $\lambda_{h-2}=\rho_{h-2}$, we have:
	\begin{eqnarray} \label{eq:h-complete}
	\eta_{h}= \omega_{h-1} + \max\{\omega_{h-2}, (\omega_{h-2}+1)/2 + \eta_{h-2}\}.
	\end{eqnarray} 
	
	We now inductively prove that $\omega_{h}\leq k\cdot c^{h} -1$ and $\eta_{h}\leq \alpha\cdot k\cdot c^{h} -1$, for any $h=1,2,\dots$, where $\alpha$, $k$, and $c$ are suitable constants (to be determined) such that $1/2<\alpha<1$, $k>1$, and $1<c<2$; in particular, we would like $c$ to be as small as possible. If $h=1$ or $h=2$, then, for any constants $c>1$ and $1/2<\alpha<1$, a constant $k$ can be chosen large enough so that $\omega_{h}\leq k\cdot c^{h} -1$ and $\eta_{h}\leq \alpha\cdot k\cdot c^{h} -1$. Indeed, it suffices to choose $k\geq 4/(\alpha\cdot c)$ in order to have $k\cdot c^{h} -1> \alpha\cdot k\cdot c^{h} -1 \geq 3\geq \omega_{1},\eta_{1},\omega_{2},\eta_{2}$.
	
	For the inductive case, assume that $\omega_{h'}\leq k\cdot c^{h'} -1$ and $\eta_{h'}\leq \alpha\cdot k\cdot c^{h'} -1$, for every integer $h'<h$; we prove the same inequalities for $h$. 
	
	By applying induction in equation~(\ref{eq:h-complete}), we get $\eta_{h}\leq  k\cdot c^{h-1} -1 + \max\{k\cdot c^{h-2} -1, (k\cdot c^{h-2} -1 +1)/2 + \alpha\cdot k\cdot c^{h-2} -1 \}\leq k\cdot c^{h-1} -1 + k\cdot c^{h-2}\cdot \max\{1, \alpha + 1/2 \}\leq k\cdot c^{h-1} -1 + (\alpha + 1/2)\cdot k\cdot c^{h-2}$, where we exploited $\alpha>1/2$. Hence, we want $k\cdot c^{h-1} -1 + (\alpha + 1/2)\cdot k\cdot c^{h-2} \leq \alpha\cdot k\cdot c^{h} -1$, that is
	\begin{eqnarray} \label{eq:c-complete-bound1}
	\alpha\cdot c^2-c-(\alpha + 1/2)\geq 0.
	\end{eqnarray}
	
	In order to establish $\omega_{h}\leq k\cdot c^{h} -1$ we distinguish two cases, based on which term determines the maximum in equation~(\ref{eq:w-complete}). 
	
	If $\omega_{h-2} + 2\eta_{h-2} \geq 2\eta_{h-1}+1$, then $\omega_{h}= \omega_{h-2} + 2\eta_{h-2}$. By applying induction, we get $	\omega_{h}\leq k\cdot c^{h-2} -1 + 2(\alpha \cdot k\cdot c^{h-2} -1)=(2\alpha+1)\cdot k\cdot c^{h-2}-3$. Hence, we want $(2\alpha+1)\cdot k\cdot c^{h-2}-3 \leq k\cdot c^{h} -1$, which is true as long as $c^2\geq 2\alpha+1$, hence
	\begin{eqnarray} \label{eq:c-complete-bound2}
	c\geq \sqrt{2\alpha + 1}.
	\end{eqnarray}
	
	If $2\eta_{h-1}+1 > \omega_{h-2} + 2\eta_{h-2}$, then $\omega_{h}= 2\eta_{h-1}+1$. By applying induction, we get $\omega_{h}\leq  2(\alpha \cdot k\cdot c^{h-1} -1) + 1 =2\alpha \cdot k\cdot c^{h-1}-1.$ Hence, we want $2\alpha \cdot k\cdot c^{h-1}-1 \leq k\cdot c^{h} -1$, which is true as long as 
	\begin{eqnarray} \label{eq:c-complete-bound3}
	c\geq 2\alpha.
	\end{eqnarray}
	
	Now pick $c=2\alpha$, thus satisfying inequality (\ref{eq:c-complete-bound3}). Substituting $c=2\alpha$ in inequality (\ref{eq:c-complete-bound1}), we want $4\alpha^3 -3\alpha-1/2\geq 0$, which is true if $\alpha\geq 0.9397$. So take $\alpha= 0.9397$, implying that inequality (\ref{eq:c-complete-bound1}) is satisfied and note that $c=2\alpha = 1.8794>1.6969>\sqrt{2\alpha+1}$, hence inequality (\ref{eq:c-complete-bound2}) is satisfied as well. This completes the induction and hence proves that $\omega_{h},\eta_{h} \in O(1.8794^h)$. 
	
	Since $h \in O(\log_3 n)$, we have $\omega_{h},\eta_{h} \in O(n^{\log_3 1.8794})\in O(n^{0.5744})$. The theorem follows.
\end{proof}

\subsection{1-2 Drawings vs Drawings with the Subtree Separation Property} \label{se:12-optimal}

In this section we show that 1-2 drawings require minimum area among the planar straight-line orthogonal drawings satisfying the subtree separation property. This is proved in the following lemma. 

\begin{lemma} \label{le:separation}
	Suppose that the complete ternary tree $T_h$ has a planar straight-line orthogonal drawing $\Gamma$ with the subtree separation property, with left width $\lambda$, with right width $\rho$, with height $\eta$, and such that, if $h\geq 2$, then the three children of $r(T_h)$ are to the left, below, and to the right of $r(T_h)$.  
	
	Then there is a 1-2 drawing $\Gamma'$ of $T_h$ with the following properties:
	\begin{itemize}
		\item the left and right widths of $\Gamma'$ are both equal to a value $\mu\leq \min\{\lambda,\rho\}$;
		\item the height of $\Gamma'$ is at most $\eta$;  and
		\item if $h\geq 2$, then let $L$, $B$, and $R$ be the subtrees of $r(T_h)$ rooted at the children of  $r(T_h)$ to the left, below, and to the right of $r(T_h)$, respectively; then the 1-2 drawings of $L$ and $R$ in $\Gamma'$ are congruent, up to a rotation of $180\degree$.     
	\end{itemize}
\end{lemma}

\begin{proof}
	We prove the statement by induction on $h$. The statement is trivially true if $h=1$, as the only (up to translations) planar straight-line orthogonal drawing of $T_1$ is a 1-2 drawing. 
	
Assume next that $h\geq 2$. We perform a sequence of modifications to $\Gamma$, eventually resulting in a 1-2 drawing $\Gamma'$ as required by the statement.  
	
First, we replace the drawings $\Gamma_L$, $\Gamma_B$, and $\Gamma_R$ of $L$, $B$, and $R$ in $\Gamma$ by 1-2 drawings of $L$, $B$, and $R$, respectively; refer to Fig.~\ref{fig:replacement}(a). Construct a copy $\Lambda_L$ of $\Gamma_L$ and rotate it counterclockwise by $90\degree$; the children of $r(L)$, if any, are to the left, below, and to the right of $r(L)$ in $\Lambda_L$. Inductively construct a 1-2 drawing $\Gamma'_L$ of $L$ such that the left and right widths of $\Gamma'_L$ are both equal to a value $\nu$ which is smaller than or equal to the minimum between the left and right widths of $\Lambda_L$, and such that the height of $\Gamma'_L$ is smaller than or equal to the height of $\Lambda_L$. Rotate $\Gamma'_L$ clockwise by $90\degree$. We now replace $\Gamma_L$ by $\Gamma'_L$ so that the rightmost vertical line intersecting $\Gamma'_L$ coincides with the rightmost vertical line intersecting $\Gamma_L$ and so that $r(L)$ is on the horizontal line through $r(T_h)$. This replacement does not increase (and possibly decreases) the left width and the height of $\Gamma$. The right width of $\Gamma$ is not altered by the modification. Note that, after the rotation and the replacement, the top and bottom heights of $\Gamma'_L$ are equal.
The replacements of $\Gamma_B$ and $\Gamma_R$ by 1-2 drawings $\Gamma'_B$ and $\Gamma'_R$ are performed similarly (actually the replacement of $\Gamma_B$ by $\Gamma'_B$ does not require any rotations).
	
	\begin{figure}[tb]
		\begin{center}
			\begin{tabular}{c c c}
				\mbox{\includegraphics[scale = 0.9]{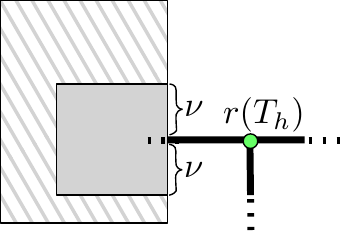}} & \hspace{4mm}
				\mbox{\includegraphics[scale = 0.9]{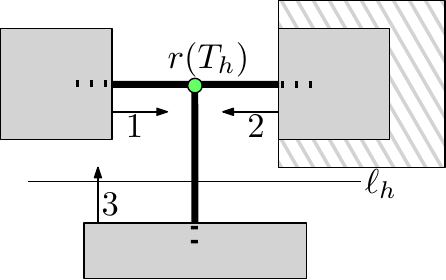}} & \hspace{4mm}
				\mbox{\includegraphics[scale = 0.9]{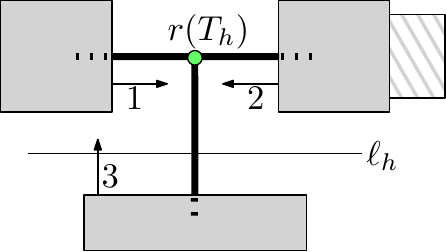}} \\
				(a) & \hspace{4mm}(b) & \hspace{4mm}(c)	
			\end{tabular}
			\caption{(a) Replacement of the drawing $\Gamma_L$, which is represented by a shaded rectangle, by a (rotated) 1-2 drawing $\Gamma'_L$, which is represented by a filled rectangle. (b)-(c) Modifications to $\Gamma$ in Case~1, assuming that the width of $\Gamma'_L$ is smaller than or equal to the width of $\Gamma'_R$. The drawings $\Gamma'_R$ and $\Gamma''_R$ are represented by a shaded and a filled rectangle, respectively. The case in which the height of $\Gamma'_L$ is smaller than the height of $\Gamma'_R$ and the case in which the height of $\Gamma'_L$ is larger than or equal to the height of $\Gamma'_R$ are represented in (b) and (c), respectively. The arrows show the final translations of the drawings of $L$, $R$, and $B$.} 
			\label{fig:replacement}
		\end{center}
	\end{figure}
	
	Since $\Gamma$ (even after the above modification) satisfies the subtree separation property and since $r(L)$ and $r(B)$ are to the left and below $r(T_h)$, respectively, there is a horizontal line (not necessarily a grid row) having $\Gamma'_L$ above and $\Gamma'_B$ below, or there is a vertical line (not necessarily a grid column) having $\Gamma'_L$ to the left and $\Gamma'_B$ to the right. Possibly both such lines exist. Analogously, there is a horizontal line having $\Gamma'_R$ above and $\Gamma'_B$ below, or there is a vertical line having $\Gamma'_R$ to the right and $\Gamma'_B$ to the left. We distinguish four cases.
	
	{\em Case 1: There is a horizontal line $\ell_h$ having $\Gamma'_L$ and $\Gamma'_R$ above and $\Gamma'_B$ below.} We first make the drawings of $L$ and $R$ congruent, up to a rotation of $180\degree$. This is done as follows (refer to Figs.~\ref{fig:replacement}(b) and~\ref{fig:replacement}(c)). Assume that the width of $\Gamma'_L$ is smaller than or equal to the width of $\Gamma'_R$; the case in which the width of $\Gamma'_R$ is smaller than the width of $\Gamma'_L$ is dealt with symmetrically. We construct a copy of $\Gamma'_L$ and rotate it by $180\degree$; denote by $\Gamma''_R$ the resulting drawing. We replace $\Gamma'_R$ by $\Gamma''_R$, so that so that the leftmost vertical line intersecting $\Gamma''_R$ coincides with the leftmost vertical line intersecting $\Gamma'_R$ and so that $r(R)$ is on the horizontal line through $r(T_h)$. Because of the assumption on the widths of $\Gamma'_L$ and $\Gamma'_R$, this modification does not increase the width of $\Gamma$. Further, since the top and bottom heights of $\Gamma'_L$ are equal and the same is true for the top  and bottom heights of $\Gamma'_R$, the modification does not increase the height of $\Gamma$. In particular, if the height of $\Gamma'_L$ is smaller than the height of $\Gamma'_R$, as in Fig.~\ref{fig:replacement}(b), then the height of $\Gamma$ decreases, while if the height of $\Gamma'_L$ is larger than or equal to the height of $\Gamma'_R$, as in Fig.~\ref{fig:replacement}(c), then the height of $\Gamma$ stays unchanged.  
	
	Finally, we act on the distances between the drawings of $L$, $R$, and $B$. We translate the drawing of $L$ rightwards, so that the rightmost vertical line intersecting it is one unit to the left of $r(T_h)$, we translate the drawing of $R$ leftwards, so that the leftmost vertical line intersecting it is one unit to the right of $r(T_h)$, and we translate the drawing of $B$ upwards, so that the topmost horizontal line intersecting it is one unit below the bottommost horizontal line intersecting the drawing of $L$. This results in a 1-2 drawing $\Gamma'$ with the required properties; note that $\Gamma'$ arranges 1-2 drawings of $L$, $B$, and $R$ as in Construction~2.
	
	{\em Case 2: There are a vertical line $\ell_v$ having $\Gamma'_L$ to the left and $\Gamma'_B$ to the right and a vertical line $\ell'_v$ having $\Gamma'_B$ to the left and $\Gamma'_R$ to the right.} Refer to Fig.~\ref{fig:replacement-2}(a). We first make the drawings of $L$ and $R$ congruent, up to a rotation of $180\degree$, as in Case~1. We now translate the drawing of $B$ upwards, so that the topmost horizontal line intersecting it is one unit below $r(T_h)$, we then translate the drawing of $L$ rightwards, so that the rightmost vertical line intersecting it is one unit to the left of the leftmost vertical line intersecting the drawing of $B$, and we finally translate the drawing of $R$ leftwards, so that the leftmost vertical line intersecting it is one unit to the right of the rightmost vertical line intersecting the drawing of $B$. This results in a 1-2 drawing $\Gamma'$ with the required properties; note that $\Gamma'$ arranges 1-2 drawings of $L$, $B$, and $R$ as in Construction~1.
	
	\begin{figure}[tb]
		\begin{center}
			\begin{tabular}{c c c}
				\mbox{\includegraphics[scale = 0.9]{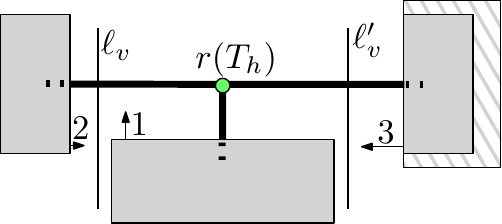}} & \hspace{4mm}
				\mbox{\includegraphics[scale = 0.9]{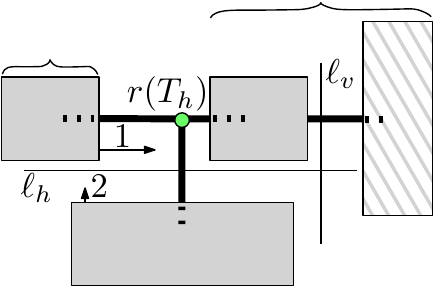}} & \hspace{4mm}
				\mbox{\includegraphics[scale = 0.9]{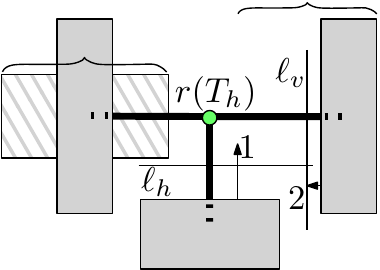}} \\
				(a) & \hspace{4mm}(b) & \hspace{4mm}(c)	
			\end{tabular}
			\caption{(a) Modifications to $\Gamma$ in Case 2, assuming that the width of $\Gamma'_L$ is smaller than or equal to the width of $\Gamma'_R$. (b) Modifications to $\Gamma$ in Case 3, assuming that the width of $\Gamma'_L$ is smaller than or equal to the right width of $\Gamma$. (c) Modifications to $\Gamma$ in Case 3, assuming that the width of $\Gamma'_L$ is larger than the right width of $\Gamma$.} 
			\label{fig:replacement-2}
		\end{center}
	\end{figure}
	
	{\em Case 3: There are a horizontal line $\ell_h$ having $\Gamma'_L$ above and $\Gamma'_B$ below and a vertical line $\ell_v$ having $\Gamma'_R$ to the right and $\Gamma'_B$ to the left.} In order to replace $\Gamma'_L$ or $\Gamma'_R$ with a rotated copy of the other one, we need to follow a different strategy than in Cases~1 and~2. In fact, rather than comparing the width of $\Gamma'_L$ with the width of $\Gamma'_R$, as we did in Cases~1 and~2, we now compare the width of $\Gamma'_L$ with the right width of $\Gamma$. The reason for this is that, if the width of $\Gamma'_R$ is smaller than the width of $\Gamma'_L$, we cannot just replace $\Gamma'_L$ with a rotated copy $\Gamma''_L$ of $\Gamma'_R$ so that the rightmost vertical lines intersecting $\Gamma'_L$ and $\Gamma''_L$ coincide, as this might create intersections between $\Gamma''_L$ and $\Gamma'_B$. 
	
	If the width of $\Gamma'_L$ is smaller than or equal to the right width of $\Gamma$, as in Fig.~\ref{fig:replacement-2}(b), then we construct a copy of $\Gamma'_L$ and rotate it by $180\degree$; denote by $\Gamma''_R$ the resulting drawing. We replace $\Gamma'_R$ by $\Gamma''_R$, so that the leftmost vertical line intersecting $\Gamma''_R$ is one unit to the right of $r(T_h)$ and so that $r(R)$ is on the horizontal line through $r(T_h)$. Because of the assumption on the width of $\Gamma'_L$ and on the right width of $\Gamma$, this modification does not increase the width of $\Gamma$. Note that $\Gamma''_R$ does not cross $\Gamma'_B$ since $\ell_h$ has $\Gamma'_L$ above and $\Gamma'_B$ below and since the top and bottom heights of $\Gamma'_L$ are equal. The construction of $\Gamma'$ is completed by translating the drawings of $L$ and $B$ as in~Case~1. 
	
	If the width of $\Gamma'_L$ is larger than the right width of $\Gamma$, as in Fig.~\ref{fig:replacement-2}(c), then we construct a copy of $\Gamma'_R$ and rotate it by $180\degree$; denote by $\Gamma''_L$ the resulting drawing. We replace $\Gamma'_L$ by $\Gamma''_L$, so that the rightmost vertical line intersecting $\Gamma''_L$ is one unit to the left of the leftmost vertical line intersecting $\Gamma'_B$ and so that $r(L)$ is on the horizontal line through $r(T_h)$. Because of the assumption on the width of $\Gamma'_L$ and on the right width of $\Gamma$, and since the left width and the right width of $\Gamma'_B$ are equal, we have that this modification does not increase the width of $\Gamma$. The construction of $\Gamma'$ is completed by translating the drawings of $B$ and $R$ as in Case~2. 
	
	{\em Case 4: There are a horizontal line $\ell_h$ having $\Gamma'_R$ above and $\Gamma'_B$ below and a vertical line $\ell_v$ having $\Gamma'_L$ to the left and $\Gamma'_B$ to the right.} This case can be discussed symmetrically to Case 3. This concludes the proof of the lemma.
\end{proof}

Lemma~\ref{le:separation} immediately implies the following.

\begin{theorem} \label{th:separation}
	For any positive integer $h$, there is a 1-2 drawing of the complete ternary tree $T_h$ with height $h$ achieving minimum area among all the planar straight-line orthogonal drawings of $T_h$ satisfying the subtree separation property.  
\end{theorem}

\begin{proof}
Consider any planar straight-line orthogonal drawing $\Gamma$ of $T_h$ satisfying the subtree separation property. If the children of $r(T_h)$ are to the left, below, and to the right of $r(T_h)$ in $\Gamma$, then Lemma~\ref{le:separation} can be applied in order to construct a 1-2 drawing $\Gamma'$ of $T_h$ whose width and height are at most equal to the width and height of $\Gamma$, respectively. If the children of $r(T_h)$ are not to the left, below, and to the right of $r(T_h)$ in $\Gamma$, then $\Gamma$ can be rotated clockwise by $90\degree$, $180\degree$, or $270\degree$ so that they are, and then Lemma~\ref{le:separation} can be applied in order to construct a 1-2 drawing $\Gamma'$ of $T_h$ whose width and height are at most equal to the width and height of the rotated drawing $\Gamma$.
\end{proof}

\subsection{Computing Minimum-Area 1-2 Drawings} \label{se:experimental}

Theorem~\ref{th:separation} motivates the study of 1-2 drawings. In this section we perform an experimental evaluation of the area requirements of 1-2 drawings. This is mainly possible due to the following theorem.

\begin{theorem} \label{th:minimum-area}
A minimum-area 1-2 drawing of a complete ternary tree can be computed in polynomial time.
\end{theorem}

The proof of Theorem~\ref{th:minimum-area} is based on the following strategy\footnote{We claim that Theorem~\ref{th:minimum-area} can be generalized to ternary trees that are not necessarily complete. However, since our main interest in 1-2 drawings comes from the study of the area requirements of complete ternary trees, we opted for keeping the exposition simple and present the theorem and its proof for complete ternary trees only.}, which resembles the approach proposed in~\cite{fpr-lrd-17} in order to compute minimum-area {\em LR-drawings} of binary trees.

For any value of $h$ we aim at computing all the {\em Pareto-optimal} width-height pairs $(\omega,\eta)$ for the 1-2 drawings of $T_h$; these are the pairs such that: (i) $T_h$ admits a 1-2 drawing with width $\omega$ and height $\eta$; and (ii) there exists no pair $(\omega',\eta')$ such that $T_h$ admits a 1-2 drawing with width $\omega'$ and height $\eta'$, where $\omega'\leq \omega$, $\eta'\leq \eta$, and at least one of these inequalities is strict.

In the following lemma, we bound the number of Pareto-optimal width-height pairs.

\begin{lemma} \label{le:number-optimal}
There are $O(n)$ Pareto-optimal width-height pairs for the 1-2 drawings of a complete ternary tree with $n$ nodes. 
\end{lemma}

\begin{proof}
The statement comes from the following two observations. 

First, for any integer value $\omega$ there is at most one Pareto-optimal pair $(\omega,\eta)$. 

Second, any Pareto-optimal pair $(\omega,\eta)$ has $\omega\leq n$, as any 1-2 drawing with width greater than $n$ has a grid column not containing any vertex; then the part of the drawing to the right of such a grid column can be moved one unit to the left, resulting in a 1-2 drawing with the same height and with smaller width. 
\end{proof}

Next, we show how to compute all the Pareto-optimal width-height pairs in polynomial time.

\begin{lemma} \label{le:computing-optimal}
The Pareto-optimal width-height pairs for the 1-2 drawings of $T_h$ can be computed in polynomial (in the number of nodes of $T_h$) time. 
\end{lemma}

\begin{proof}
In this proof by {\em pair} we always mean Pareto-optimal width-height pair. Suppose that the pairs for the 1-2 drawings of $T_{h-1}$ have been computed already (note that $(1,1)$ is the only pair for the 1-2 drawings of $T_1$). We compute the pairs for the 1-2 drawings of $T_h$ by considering all the  possible triples $(p_l,p_b,c)$ where $p_l=(\omega_l,\eta_l)$ and $p_b=(\omega_b,\eta_b)$ are pairs for the 1-2 drawings of $T_{h-1}$ and $c$ is either $1$ or $2$. By Lemma~\ref{le:number-optimal} there are $O(n)$ pairs for the 1-2 drawings of $T_{h-1}$, hence there are $O(n^2)$ such triples. For each triple $(p_l,p_b,c)$ we consider the 1-2 drawing $\Gamma_h$ defined as follows: 
\begin{itemize}
	\item $\Gamma^a_{h-1}$ is a 1-2 drawing with width $\omega_b$ and height $\eta_b$; 
	\item $\Gamma^b_{h-1}$ is a 1-2 drawing with width $\omega_l$ and height $\eta_l$, clockwise rotated by $90\degree$; note that, after the rotation, $\Gamma^b_{h-1}$ has width $\eta_l$ and height $\omega_l$; 
	\item $\Gamma^c_{h-1}$ is a 1-2 drawing with width $\omega_l$ and height $\eta_l$, counterclockwise rotated by $90\degree$; note that, after the rotation, $\Gamma^c_{h-1}$ has width $\eta_l$ and height $\omega_l$;
	\item $\Gamma^a_{h-1}$, $\Gamma^b_{h-1}$, and $\Gamma^c_{h-1}$ are arranged as in Construction~$c$. 
\end{itemize}

The width and height of $\Gamma_h$ can be computed in $O(1)$ time by Properties~\ref{pr:construction-1} and~\ref{pr:construction-2}.

Out of the $O(n^2)$ 1-2 drawings of $T_h$ constructed as above we only keep the $O(n)$ drawings corresponding to Pareto-optimal width-height pairs -- it comes again from Lemma~\ref{le:number-optimal} that there are this many Pareto-optimal width-height pairs. This can be accomplished in polynomial time by ordering the $O(n^2)$ width-height pairs by increasing width and, secondarily, by increasing height, and by removing every width-height pair that is preceded by a pair with smaller or equal height. 

The correctness of the described algorithm follows by Lemma~\ref{le:separation}. In particular, the algorithm uses, in every constructed drawing of $T_h$, two drawings for the left and right subtrees of $r(T_h)$ that are congruent, up to a rotation of $180\degree$, which can be done without loss of generality by Lemma~\ref{le:separation}. Further, the algorithm constructs the Pareto-optimal pairs for the 1-2 drawings of $T_h$ starting from the Pareto-optimal pairs for the 1-2 drawings of $T_{h-1}$. This is also not a loss of generality, as any 1-2 drawing of $T_{h-1}$ which does not correspond to a Pareto-optimal pair can be replaced by a (possibly rotated) 1-2 drawing of $T_{h-1}$ which corresponds to a Pareto-optimal pair, without increasing the width and height of the drawing of $T_h$; the existence of such a  drawing again follows by Lemma~\ref{le:separation}. 
\end{proof}

Lemmata~\ref{le:number-optimal} and~\ref{le:computing-optimal} imply Theorem~\ref{th:minimum-area}, as the minimum area for a 1-2 drawing of $T_h$ is equal to $\min\{\omega \cdot \eta \}$, where the minimum is taken over all the Pareto-optimal width-height pairs $(\omega,\eta)$ for the 1-2 drawings of $T_h$.

We run a mono-thread C implementation of the algorithm that computes the Pareto-optimal width-height pairs for the 1-2 drawings of $T_h$ described in the proof of Lemma~\ref{le:computing-optimal} on a machine with two $4$-core $3.16$GHz Intel(R) Xeon(R) CPU X$5460$ processors, with $48$GB of RAM, running Ubuntu $14.04.2$ LTS. We could compute the Pareto-optimal width-height pairs $(\omega,\eta)$ and the minimum area for the 1-2 drawings of $T_h$ with $h$ up to $20$. The computation of the pairs for $h=20$ took roughly $5$ days. The table below shows the value of $h$, the corresponding value of $n$, which is $(3^h-1)/2$, and the minimum area required by any 1-2 drawing (and by Theorem~\ref{th:separation} by any planar straight-line orthogonal drawing satisfying the subtree separation property) of $T_h$.  

\definecolor{Gray}{gray}{0.9}

{\centering
\begin{table}[htb]
\centering
	\scriptsize
	\begin{tabular} {|c|c|c|c|c|c|c|c|c|c|c|} 
		\hline
		\rowcolor{Gray}{ $h$} & 1  & 2  & 3  & 4  & 5  & 6  & 7  & 8  & 9  & 10 \\
		\hline
		{ $n$} & 1 & 4 & 13 & 40 & 121 & 364 & \num{1093} & \num{3280} & \num{9841} & \num{29524}\\
		\hline
		{Area} & 1 & 6 & 25 & 99 & 342 & \num{1184} & \num{4030} & \num{13320} & \num{44457} & \num{144690} \\
		\hline
		\hline
		\rowcolor{Gray}{ $h$} & 11  & 12  & 13  & 14 &15 & 16 & 17  & 18  & 19  & 20\\
		\hline
		{ $n$} & \num{88573} & \num{265720} & \num{797161} & \num{2391484} & \num{7174453} & \num{21523360} & \num{64570081} & \num{193710244} & \num{581130733} & \num{1743392200}\\
		\hline
		{Area} & \num{469221} & \num{1520189} & \num{4840478} & \num{15550542} & \num{49461933} & \num{157388427} & \num{498895215} & \num{1580110511} & \num{4990796080} & \num{15765654805} \\
		\hline		\end{tabular}
	\label{ta:evaluation}
\end{table}
}

By means of the Mathematica software~\cite{mat} we searched for the function of the form $a\cdot n^b + c$ that better fits the values of the table above, according to the least squares optimization method. The optimal function we got is $3.3262 \cdot  x^{1.047}-\num{181209.1337}$. While the large absolute value of the additive constant suggests the need for a lower order term or for a different optimization method, the experimentation also seems to indicate that planar straight-line orthogonal drawings with the subtree separation property cannot be constructed within almost-linear area. We hence state the following conjecture.

\begin{conjecture}	
There exists a constant $\varepsilon>0$ such that $n$-node complete ternary trees require $\Omega(n^{1+\epsilon})$ area in any planar straight-line orthogonal drawing satisfying the subtree separation property.
\end{conjecture}	

\section{Lower Bound for the Minimum Side Length} \label{se:lower-bound}

In this section we show that the minimum side length of any planar straight-line orthogonal drawing of a complete ternary tree with $n$ nodes is in $\Omega(n^{\log_3 \phi})\in\Omega(n^{0.438})$, where $\phi = (1+\sqrt 5)/2$ is the golden ratio. We remark that this lower bound holds true even for drawings that do not satisfy the subtree separation property. Thus, with respect to the minimum side length of a planar straight-line orthogonal drawing, binary and ternary trees are in sharp contrast. Namely, \emph{any} $n$-node binary tree admits a planar straight-line orthogonal drawing in which the minimum side length is in $O(\log n)$~\cite{cdp-no-92}. We have the following.

\begin{theorem} \label{th:minimum-side}
The minimum side length of any planar straight-line orthogonal drawing of an $n$-node complete ternary tree is in $\Omega(n^{\log_3 ((1+\sqrt 5)/2)})\in\Omega(n^{0.438})$.
\end{theorem}

\begin{proof}
Let $\Gamma$ be any planar straight-line orthogonal drawing of the complete ternary tree $T_{h}$ with height $h$, as in Fig.~\ref{fig:lower-bound}(a). One of the children of $r(T_h)$, call it $v_1$, is such that no other child of $r(T_h)$ is drawn on the line $\ell$ through $r(T_h)$ and $v_1$. Moreover, for $i=1,2,\dots,h-2$, the node $v_i$ has exactly one child $v_{i+1}$ drawn on $\ell$. Hence, in $\Gamma$ there is a root-to-leaf path with $h$ nodes that is drawn all on the same (horizontal or vertical) line $\ell$, and that is such that no other child of $r(T_h)$ is on $\ell$. We call \emph{leg} of $\Gamma$ such a path. Analogously, in $\Gamma$ there are two root-to-leaf paths with $h$ nodes each that are both drawn on the same (horizontal or vertical) line. We call \emph{left arm} and \emph{right arm} of $\Gamma$ such paths, so that the left arm, the leg, and the right arm of $\Gamma$ occur in this counterclockwise order around $r(T_h)$. Denote by $\gamma(\Gamma)$, $\lambda(\Gamma)$, and $\rho(\Gamma)$ the length of the leg of $\Gamma$, the length of the left arm of $\Gamma$, and the length of the right arm of $\Gamma$, respectively. 

	\begin{figure}[htb]
		\begin{center}
			\begin{tabular}{c c c}
				\mbox{\includegraphics[scale = 0.9]{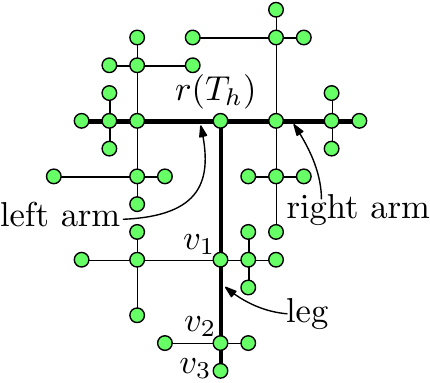}} & \hspace{4mm}
				\mbox{\includegraphics[scale = 0.9]{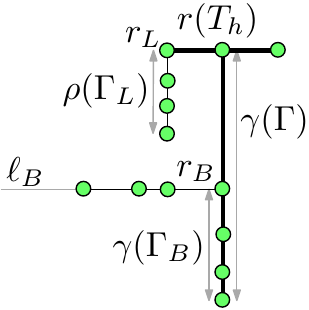}} & \hspace{4mm}
				\mbox{\includegraphics[scale = 0.9]{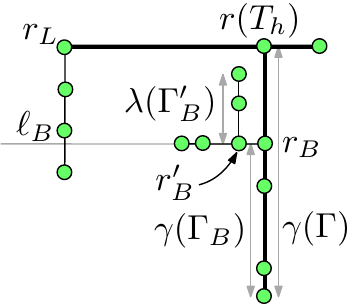}} \\
				(a) & \hspace{4mm}(b) & \hspace{4mm}(c)	
			\end{tabular}
			\caption{(a) Leg and arms (shown by thick lines) in a planar straight-line orthogonal drawing of a complete ternary tree $T_h$. (b) The length of the leg of $\Gamma$ is greater than or equal to $f(h-1)+f(h-2)$ if $\ell_B$ lies below the right arm of $\Gamma_L$. (c) The length of the leg of $\Gamma$ is greater than or equal to $f(h-1)+f(h-2)$ if $\ell_B$ intersects the right arm of $\Gamma_L$.} 
			\label{fig:lower-bound}
		\end{center}
	\end{figure}

We now define the function $f(h)=\min_{\Gamma}\{\gamma(\Gamma),\lambda(\Gamma),\rho(\Gamma)\}$, where the minimum is taken over all the planar straight-line orthogonal drawings $\Gamma$ of $T_{h}$. Note that $f(h)\geq h$, as every planar straight-line orthogonal drawing of $T_h$ has the leg, the left arm, and the right arm of length at least $h$. 

We claim that, for any $h\geq 3$, we have $f(h)\geq f(h-1)+f(h-2)$. The claim implies the theorem; this comes from the following two observations.

First, consider any planar straight-line orthogonal drawing $\Gamma$ of $T_{h}$. If $r(T_h)$ has two children on the same horizontal line in $\Gamma$, then the width of $\Gamma$ is larger than or equal to $\lambda(\Gamma)$ and the height of $\Gamma$ is larger than or equal to $\gamma(\Gamma)$. Otherwise, $r(T_h)$ has two children on the same vertical line in $\Gamma$, and then the height of $\Gamma$ is larger than or equal to $\lambda(\Gamma)$ and the width of $\Gamma$ is larger than or equal to $\gamma(\Gamma)$. 

Second, since $f(h)\geq f(h-1)+f(h-2)$, we have that $f(h)$ grows asymptotically as the terms of the Fibonacci sequence; it is known that the ratio between two consecutive terms of the Fibonacci sequence tends to the golden ratio $\phi=(1+\sqrt 5)/2\approx 1.618$. Hence $f(h) \in \Omega (\phi^h) \in \Omega (n^{\log_3 \phi})\in \Omega(n^{0.438})$. Formally, it can be proved by induction that $f(h)\geq k\cdot \phi^h$, for a constant $k>0$ that is sufficiently small so that $f(h)\geq k\cdot \phi^h$ is verified for $h=1$ and $h=2$. By induction, $f(h-1)+f(h-2)\geq k\cdot \phi^{h-1} + k\cdot \phi^{h-2}$. Hence, we need to prove that  $k\cdot\phi^{h-1} + k\cdot\phi^{h-2}\geq k\cdot\phi^{h}$, that is, $\phi^2-\phi-1\leq 0$. Since $\phi=(1+\sqrt 5)/2$ is one of the solutions to the equation $\phi^2-\phi-1=0$, it follows that $\phi^2-\phi-1\leq 0$, and hence the induction is completed and $f(h)\geq k\cdot \phi^h \in \Omega(\phi^h)$. 

It remains to prove the claim: for any $h\geq 3$, we have $f(h)\geq f(h-1)+f(h-2)$. Consider any planar straight-line orthogonal drawing $\Gamma$ of $T_{h}$. Assume that the children of $r(T_h)$ are to the left, below, and to the right of $r(T_h)$ in $\Gamma$; the other three cases can be discussed symmetrically. Denote by $\Gamma_L$, $\Gamma_B$, and $\Gamma_R$ the drawings in $\Gamma$ of the subtrees of $r(T_h)$ that are rooted at the children $r_L$, $r_B$, and $r_R$ that are to the left, below, and to the right of $r(T_h)$, respectively. Then the arms of $\Gamma_L$ and $\Gamma_R$ lie on two vertical lines, while the arms of $\Gamma_B$ lie on a horizontal line $\ell_B$. 

We prove that $\gamma(\Gamma)\geq f(h-1)+f(h-2)$. We distinguish two cases. In the first case, shown in Fig.~\ref{fig:lower-bound}(b), $\ell_B$ lies below the right arm of $\Gamma_L$. Then $\gamma(\Gamma)\geq \gamma(\Gamma_B)+\rho(\Gamma_L)$. Since $\gamma(\Gamma_B)\geq f(h-1)$ and $\rho(\Gamma_L)\geq f(h-1)$, and since $f(h-1)\geq f(h-2)$, we have $\gamma(\Gamma)\geq  f(h-1)+f(h-2)$. In the second case, shown in Fig.~\ref{fig:lower-bound}(c), $\ell_B$ intersects the right arm of $\Gamma_L$. Let $r'_B$ be the child of $r_B$ that is to the left of $r_B$ in $\Gamma$ and let $\Gamma'_B$ be the drawing in $\Gamma$ of the subtree of $r_B$ that is rooted at $r'_B$. Then $\gamma(\Gamma)\geq \gamma(\Gamma_B)+\lambda(\Gamma'_B)$ (observe that the sum $\gamma(\Gamma_B)+\lambda(\Gamma'_B)$ counts twice the grid line $\ell_B$, however does not count the horizontal line through $r(T_h)$). Since $\gamma(\Gamma_B)\geq f(h-1)$ and $\lambda(\Gamma'_B)\geq f(h-2)$, we have $\gamma(\Gamma) \geq f(h-1)+f(h-2)$.

Next, we prove that $\lambda(\Gamma)\geq f(h-1)+f(h-2)$. We again distinguish two cases. If $\ell_B$ intersects the right arm of $\Gamma_L$, then $\lambda(\Gamma)\geq \lambda(\Gamma_B)+\gamma(\Gamma_L)$. Since $\lambda(\Gamma_B)\geq f(h-1)$ and $\gamma(\Gamma_L)\geq f(h-1)$, and since $f(h-1)\geq f(h-2)$, we have $\lambda(\Gamma)\geq f(h-1)+f(h-2)$. Otherwise, $\ell_B$ lies below the right arm of $\Gamma_L$. Let $r'_{L}$ be the child of $r_L$ that is below $r_L$ in $\Gamma$ and let $\Gamma'_L$ be the drawing in $\Gamma$ of the subtree of $r_L$ that is rooted at $r'_L$. Then $\lambda(\Gamma)\geq \gamma(\Gamma_L)+\rho(\Gamma'_L)$. Since $\gamma(\Gamma_L)\geq f(h-1)$ and $\rho(\Gamma'_L)\geq f(h-2)$, we have $\lambda(\Gamma) \geq f(h-1)+f(h-2)$.

Finally, the proof that $\rho(\Gamma)\geq f(h-1)+f(h-2)$ is symmetric to the proof that $\lambda(\Gamma)\geq f(h-1)+f(h-2)$.
\end{proof}

\section{Conclusions and Open Problems} \label{se:conclusions}

In this paper we studied the area requirements of planar straight-line orthogonal drawings of ternary trees. Several problems related to this topic remain open.

\begin{open} \label{op1}
Let $f(n)$ be the smallest function such that every $n$-node ternary tree admits a planar straight-line orthogonal drawing in area $f(n)$. What is the (asymptotic) value of $f(n)$?
\end{open}

We proved the first sub-quadratic area upper bound for $f(n)$; namely, our bound is \generalArea. For complete ternary trees better area bounds can be achieved, however the following is also open. 

\begin{open} \label{op2}
	Let $g(n)$ be the smallest function such that an $n$-node complete ternary tree admits a planar straight-line orthogonal drawing in area $g(n)$. What is the (asymptotic) value of $g(n)$?
\end{open}

We suspect that $f(n)$ and $g(n)$ are super-linear functions of $n$. In particular, motivated by an extensive experimental analysis, we conjectured in this paper that $g(n)$ (and hence $f(n)$) is in $\Omega(n^{1+\epsilon})$, for some suitable constant $\epsilon$, if one restricts the attention to drawings that satisfy the subtree separation property. 

It is interesting that, differently from binary trees, $O(n \textrm{ polylog} (n))$ area upper bounds cannot be achieved by ``squeezing'' the drawings in one direction only. Indeed, we proved that an $n$-node complete ternary tree requires polynomial width {\em and} height in any planar straight-line orthogonal drawing; our lower bound for the minimum side length of a planar straight-line orthogonal drawing of an $n$-node complete ternary tree is $\Omega(n^{\log_3 \phi})\in \Omega(n^{0.438})$, where $\phi$ is the golden ratio. As a consequence, the following problem seems central to the study of the area requirements of planar straight-line orthogonal drawings of ternary trees.

\begin{open} \label{op3}
	Does every $n$-node ternary tree admit a planar straight-line orthogonal drawing in which both the width and the height are in $o(n)$?
\end{open}

Finally, it is not clear to us whether planar straight-line orthogonal drawings in which one side has length matching our $\Omega(n^{\log_3 \phi})$ lower bound can be achieved for all $n$-node ternary trees. While we proved that this can be done for $n$-node complete ternary trees, for general $n$-node ternary trees our best bound is only \generalHeight. Thus, we ask the following.

\begin{open} \label{op4}
	Does every $n$-node ternary tree admit a planar straight-line orthogonal drawing in which one side has length in $O(n^{\log_3 \phi})$, where $\phi$ is the golden ratio?
\end{open}

Clearly, a positive answer to Open Problem~\ref{op4} would imply an improved upper bound for the function $f(n)$ from Open Problem~\ref{op1}. 

\bibliographystyle{splncs03} 
\bibliography{bibliography}

\begin{thebibliography}{10}
\providecommand{\url}[1]{\texttt{#1}}
\providecommand{\urlprefix}{URL }

\bibitem{a-so-15}
Ali, A.: Straight line orthogonal drawings of complete ternery trees. Tech.
  rep., {MIT} Summer Program in Undergraduate Research Final Paper (July 2015),
  {\sc https://math.mit.edu/research/undergraduate/spur/documents/2015ali.pdf}

\bibitem{c-tdr-18}
Chan, T.M.: Tree drawings revisited. In: Speckmann, B., T{\'{o}}th, C.D. (eds.)
  Symposium on Computational Geometry (SoCG '18). LIPIcs, vol.~99, pp.
  23:1--23:15. Schloss Dagstuhl - Leibniz-Zentrum fuer Informatik (2018)

\bibitem{cgkt-oa-02}
Chan, T.M., Goodrich, M.T., Kosaraju, S.R., Tamassia, R.: Optimizing area and
  aspect ratio in straight-line orthogonal tree drawings. Comput. Geom.  23(2),
   153--162 (2002)

\bibitem{cfp-arslodtt-18}
Covella, B., Frati, F., Patrignani, M.: On the area requirements of
  straight-line orthogonal drawings of ternary trees. In: Iliopoulos, C.S.,
  Leong, H.W., Sung, W. (eds.) 29th International Workshop on Combinatorial
  Algorithms ({IWOCA} 2018). LNCS, vol. 10979, pp. 128--140. Springer (2018)

\bibitem{cdp-no-92}
Crescenzi, P., {Di Battista}, G., Piperno, A.: A note on optimal area
  algorithms for upward drawings of binary trees. Comput. Geom.  2,  187--200
  (1992)

\bibitem{dett-gd-99}
{Di Battista}, G., Eades, P., Tamassia, R., Tollis, I.G.: Graph Drawing:
  Algorithms for the Visualization of Graphs. Prentice-Hall (1999)

\bibitem{df-dt-13}
{Di Battista}, G., Frati, F.: Drawing trees, outerplanar graphs,
  series-parallel graphs, and planar graphs in a small area. In: Pach, J. (ed.)
  Thirty Essays on Geometric Graph Theory, pp. 121--165. Springer-Verlag New
  York (2013)

\bibitem{f-so-07}
Frati, F.: Straight-line orthogonal drawings of binary and ternary trees. In:
  Hong, S., Nishizeki, T., Quan, W. (eds.) Symposium on Graph Drawing ({GD}
  '07). LNCS, vol. 4875, pp. 76--87. Springer (2008)

\bibitem{fpr-lrd-17}
Frati, F., Patrignani, M., Roselli, V.: {LR}-drawings of ordered rooted binary
  trees and near-linear area drawings of outerplanar graphs. In: Klein, P.N.
  (ed.) Symposium on Discrete Algorithms ({SODA} 2017). pp. 1980--1999. {SIAM}
  (2017)

\bibitem{gr-sdbtla-02}
Garg, A., Rusu, A.: Straight-line drawings of binary trees with linear area and
  arbitrary aspect ratio. In: Kobourov, S.G., Goodrich, M.T. (eds.) Symposium
  on Graph Drawing ({GD} 2002). LNCS, vol. 2528, pp. 320--331. Springer (2002)

\bibitem{gt-ccurpt-01}
Garg, A., Tamassia, R.: On the computational complexity of upward and
  rectilinear planarity testing. {SIAM} J. Comput.  31(2),  601--625 (2001)

\bibitem{ntu-odog-05}
Nomura, K., Tayu, S., Ueno, S.: On the orthogonal drawing of outerplanar
  graphs. {IEICE} Transactions  88-A(6),  1583--1588 (2005)

\bibitem{r-tda-16}
Rusu, A.: Tree drawing algorithms. In: Tamassia, R. (ed.) Handbook of Graph
  Drawing and Visualization, chap.~5, pp. 155--192. CRC Press (2016)

\bibitem{rs-gdbt-08}
Rusu, A., Santiago, C.: Grid drawings of binary trees: An experimental study.
  J. Graph Algorithms Appl.  12(2),  131--195 (2008)

\bibitem{s-lpag-76}
Shiloach, Y.: Linear and Planar Arrangement of Graphs. Ph.D. thesis, Weizmann
  Institute of Science, Rehovot (February 1976)

\bibitem{skc-ae-00}
Shin, C., Kim, S.K., Chwa, K.: Area-efficient algorithms for straight-line tree
  drawings. Comput. Geom.  15(4),  175--202 (2000)

\bibitem{t-egg-87}
Tamassia, R.: On embedding a graph in the grid with the minimum number of
  bends. {SIAM} J. Comput.  16(3),  421--444 (1987)

\bibitem{mat}
{Wolfram Research Inc.}: Mathematica 10 (2014), \sc http://www.wolfram.com

\end{thebibliography}

\end{document}